%% file: index.tex
\documentclass[submission,copyright,creativecommons]{eptcs}
\providecommand{\event}{arXiv}

\usepackage{hyperref}
\usepackage{stmaryrd}
\usepackage{amsmath}
\usepackage{amssymb}
\usepackage{notation}
\usepackage{bussproofs}
\usepackage[ruled,vlined]{algorithm2e}

\usepackage{amsthm}
\usepackage{definitions}
\newcommand\IfExtended[1]{#1}
\newcommand\IfNotExtended[1]{}

\author{Johannes Schoisswohl 
    \institute{ University of Manchester, UK \and TU Wien, Austria }
    \email{johannes.schoisswohl@manchester.ac.uk}
   \and Laura Kovács
     \institute{ TU Wien, Austria } 
    \email{laura.kovacs@tuwien.ac.at}
   }

   \title{Automating Induction by Reflection}

\def\titlerunning{Automating Induction by Reflection}
\def\authorrunning{J. Schoisswohl \& L. Kovács}

\begin{document}
  \maketitle

  \begin{abstract}
    \input{abstract.tex}

    This paper is an extended version of the LFMTP 21 submission with the same title.
  \end{abstract}

  \section{Introduction}

\input{introduction.tex}

  
  \section{Preliminaries}
  \input{preliminaries.tex}

  \section{Reflective extension}
  \label{sect:reflExt}
  \input{reflective_extension.tex}

  \section{Induction by Reflection}%
  \label{sect:finiteInd}
  \input{induction_by_reflection.tex}

  \section{Experiments}
  \label{sect:experiments}
  \input{experiments.tex}

  \section{Conclusion}
  \input{conclusion.tex}

  \paragraph{Acknowledgements}
  This work has been supported by the ERC consolidator grant 2020
  ARTIST 101002685, the ERC starting grant 2014 SYMCAR 639270, the EPSRC grant EP/P03408X/1, and the ERC proof of concept grant 2018 SYMELS 842066.

  \bibliographystyle{eptcs}
  \bibliography{index.bib}

\end{document}

%% file: abstract.tex
Despite recent advances in automating  theorem proving in full
first-order theories, inductive reasoning still poses a serious challenge to state-of-the-art theorem provers. 
The reason for that is that in first-order logic induction
requires an infinite number of axioms, which is not a feasible
input to a computer-aided theorem prover requiring a finite input. 
Mathematical practice is to specify these infinite sets of axioms as axiom schemes. 
Unfortunately these schematic definitions cannot be formalized in
first-order logic, and therefore not supported as  inputs for
first-order  theorem provers.

In this work we introduce a new method, inspired by the field of
axiomatic theories of truth, that allows to express schematic
inductive definitions, in the standard syntax of multi-sorted first-order logic. 
Further we test the practical feasibility of the method with state-of-the-art theorem provers, comparing it to solvers' native techniques for handling induction.

%% file: introduction.tex
Automated reasoning has advanced tremendously in the last decades, pushing the limits of what computer programs can prove about other computer programs. 
Recent progress in this area features techniques such as 
first-order reasoning with term algebras~\cite{POPL17Kovacs}, 
embedding programming control structures in first-order logic~\cite{FOOLKotelnikovKRV16},
the \avatar{} architecture~\cite{AVATAR14},
combining theory instantiation and unification with abstraction in saturation-based proof search~\cite{TheoryInstUWA18},
automating proof search for higher-order logic~\cite{Zipperposition,BhayatR20a},
and first-order logic with rank-1 polymorphism~\cite{BhayatR20b}.

Despite the fact that first-order logic with equality can be handled very well in many practical cases, there is one fundamental mathematical concept most first-order theorem provers lack; namely inductive reasoning. Not only is induction a very interesting theoretical concept, it is also of high practical importance, since it is required to reason about software reliability, safety and security, see e.g.~\cite{HoderBM11,Maffei18,Sagiv19,PickFG20,KovacsFMCAD20}.
Such and similar sofwtare requirememt typically involve properties over natural numbers, recursion, unbounded loop iterations, or recursive data structures like lists, and trees.

Many different approaches towards automating inductive reasoning have been proposed, ranging from 
cyclic proofs~\cite{EchenimP20,KersaniP13}, 
computable approximations of the $\omega$-rule~\cite{BakerIrelandSmaill92},
recursion analysis~\cite{Aubin79,Cruanes17,Moore19}, 
theory exploration~\cite{HipSpec,ReynoldsK15}, 
inductive strengthening~\cite{Leino12,ReynoldsK15}, 
and integrating induction in saturation based theorem proving~\cite{InductionWithGeneralization20,Cruanes17,VampireInduction19}. 
What all these approaches have in common is that they tackle the problem of inductive theorem proving by specializing the proof system. In the present paper,  we propose a different approach. Instead of changing the proof system, we will change our input problem, replacing the infinite induction scheme by a finite conservative extension. Hence our approach is not tailored to one specific reasoner, but can be used with any automated theorem prover for first-order logic.


In order to achieve such a generic approach,  we use ideas from axiomatic theories of truth~\cite{Horsten11}.  
As G\"odel's incompleteness theorem tells us that there is a formula $Bew(x)$ in Peano Arithmetic  $\PA$
that expresses provability in $\PA$, Tarski's undefinability theorem teaches us that there is no formula $T(x)$ that expresses truth in $\PA$. 
Extending the language of $\PA$, in order to be able to express truth in $\PA$ is the core idea of axiomatic theories of truth.
The truth theory we introduce in this paper will not be an extension of \PA{}, but an extension of an arbitrary theory. 
This theory will let us express the sentence ``for all formulas, the induction scheme is true'', within first order logic.

As our experimental evaluation shows our technique does not outperform state of the art built-in methods for inductive reasoning in general, we will see that there are cases where an improvement can be achieved.

The contributions of our work can be summarized as follows:
\begin{itemize}
  \item We introduce a method for conservatively extending an arbitrary theory with a truth predicate
    (Section~\ref{sect:reflExt}); 
  \item We show how our method can be used to replace the induction scheme of \PA{}, and other theories involving inductive datatypes
    (Section~\ref{sect:finiteInd}); 
  \item We provide a new set of benchmarks to the automated theorem proving community
    (Section~\ref{sect:experiments}); 
  \item We conduct experiments on this set of benchmarks with state-of-the-art theorem provers. 
    (Section~\ref{sect:experiments})
\end{itemize}

%% file: preliminaries.tex
We assume familiarity with multi-sorted first-order logic and automated theorem proving; for details, we refer to~\cite{Vampire13,BhayatR20b}. 

We consider reasoning on the object and the meta level. Therefore we will use different symbols for logic on all of these levels.  We use the symbols 
$\onot$, 
$\oand$, 
$\oor$, 
$\oimpl$,
$\oiff$, and
$\oeq$
for negation, conjunction, disjunction, implication, equivalence, and equality respectively, and write $Q x: \sigma. \phi$ with $Q \in \s{\oexists, \oforall}$ for existential and universal quantification over the sort $\sigma$ on the object level.
We will drop sort declarations for quantified variables when there is no ambiguity. Further we will write $\universalClosure[x_1, \ldots, x_n]\phi$ for the formula $\forall x_1. \ldots \forall x_n . \phi$, and $\universalClosure\phi$ to denote the universal closure of $\phi$. 

On the meta-level we will use 
$\mnot$, 
$\mand$, 
$\mor$, 
$\mimpl$,
$\miff$, and
$=$
for negation, conjunction, disjunction, implication, equivalence, and equality and $\mforall$, and $\mexists$ for quantification. Meta-level logical formulas will only be used where they help to improve readability and precision, and otherwise natural language will be used. 

By $\ovar\sigma$, $\oterm\sigma$, and $\oform$, we respectively denote the sets of variables of sort $\sigma$, terms of sort $\sigma$ and object-level formulas over a signature $\Sigma$. As $\ovar\sigma$ is a countably infinite set, we assume without loss of generality that it is composed of the variables $\s{\x i^\sigma \mid i \in \mathbb N}$, and leave away the sort superscript $\sigma$, if it is clear from the context. 

For a function symbol $f$ from signature $\Sigma$, we write $f\softype\sigma_1 \sprod \ldots \sprod \sigma_n \sarrow \sigma \in \Sigma$, to denote that $\sdom(f) = \sigma_1\times \ldots \times\sigma_n$ is the domain of $f$, $\scodom(f) = \sigma$ is its codomain, and $\sarity(f) = n$ is the arity of $f$. We consider constants being functions of arity 0 and write $c\softype\sigma$ for $c\softype\sarrow\sigma$. Further we write $P: \spred{\sigma_n \sprod  \ldots  \sprod\sigma_n}$ to denote that $P$ is a predicate with domain $\sdom(P) =  \sigma_1\times \ldots \times\sigma_n$ and arity $\sarity(P) = n$. Further, we write $\sdom(s, i)$ to refer to the $i$th component of the domain of $s$.

Given formula/term $\phi$, a variable $x$ and a term $t$, we write $\phi[x\mapsto t]$ to denote the formula/term resulting from replacing all occurences of $x$ by $t$ in $\phi$. Similarly, if $x$ is a variable and $\phi[x]$ is a formula/term, we denote the formula/term resulting from replacing all occurences of $x$ for $t$ by $\phi[t]$.

A formula is open if it contains free variables, and closed otherwise. We consider a theory to be a set of closed formulas. If $\T$ is a theory with signature $\Sigma$, by $\Lang{\T}$ we denote the set of all formulas over $\Sigma$.

The semantics of formulas and terms over a signature $\Sigma$ is defined using multi-sorted first-order interpretations $\M$, consisting of $\tup{\tup{\domain_{\sigma_1}, \ldots, \domain_{\sigma_n}}, \I}$, where: $\domain_{\sigma_i}$ is the domain for sort $\sigma \in \ssorts[\Sigma]$, and $\I$ is an interpretation function that freely interprets variables, function symbols, and predicate symbols, respecting the sorts, and is extended to terms in the standard way. 
By $\domain$ we denote $\tup{\domain_{\sigma_1}, \ldots, \domain_{\sigma_n}}$.
We write $\M \vDash \phi$ for the structure $\M$ satisfying the formula $\phi$ and say that $\M$ is a model of $\phi$. We write $\I \vDash \phi$ instead of $\tup{\tup{\domain_{\sigma_1}}, \ldots, \domain_{\sigma_n}, \I} \vDash \phi$, whenever the domains are clear from context. By $\interpretation\Sigma$ we denote the class of all interpretation functions over a signature $\Sigma$.


For defining our approach for reflective reasoning, we need the concept of a conservative extension. A conservative extension of a theory $\T$ is a theory $\T'$, such that $\Sigma_{\T} \subseteq \Sigma_{\T'}$, and for all for $\mforall \phi \in \Lang{\T}. (\T \vDash \phi \miff \T' \vDash \phi)$

An inductive datatype $\D_\tau$, with respect to some signature $\Sigma$ is a pair $\tup{\tau, \ctors[\tau]}$, where $\tau$ is a sort and $\ctors[\tau]$ is a set of function symbols $F \subset \Sigma$ such that $\mforall f \in F. \scodom(f) = \tau$.  By the first-order structural induction scheme of $\D_\tau$ we denote the set of formulas:

\begin{align}
  \tag{$\foInd{ \tau }$}\Big\{  \big(\bigwedge_{c \in \ctors[\tau]} case_{c}\big)  \oimpl \oforall x .\phi[x] \mid \phi[x] \in \oform \Big \}
\end{align}
where \begin{align*}
  case_c &= \universalClosure[x_1, ..., x_n] \Big(  \big(\bigwedge_{i \in recursive_c }\phi[ x_i ])  \oimpl \phi[ c(x_1, ..., x_n) ]\Big)\\
  recursive_c &= \s{ i \mid \sdom[\Sigma](c,i) = \tau }
\end{align*}

$x$ the induction variable,
$\bigwedge_{c \in \ctors[\tau]} case_{c}$ the induction premise, and 
$\oforall x \phi[ x ]$ the induction conclusion.

An example for such an inductive datatype is the type of lists $\D_{\mathsf{List}} = \tup{\mathsf{List}, \s{\mathsf{nil} \softype \mathsf{List}, \mathsf{cons} \softype \alpha \sprod \mathsf{List} \sarrow \mathsf{List}}}$. The first-order induction scheme for $\D_{\mathsf{List}}$ is therefore 

\begin{align*}
  \Big\{  case_{\mathsf{nil}} \land case_{\mathsf{cons}}  &\oimpl \oforall x .\phi[x] \mid \phi[x] \in \oform \Big \}
   & case_{\mathsf{nil}} &= \top \oimpl \phi[\mathsf{nil}]\\
  && case_{\mathsf{cons}} &= \forall x:\alpha, xs:\mathsf{List}. \Big( \phi[xs] \oimpl \phi[\mathsf{cons}(x,xs)]\Big)\\
\end{align*}

%% file: reflective_extension.tex
Our aim is to finitely axiomatise the induction scheme $\foInd{ \tau }$ with respect to some datatypes $\D$ and an arbitrary base theory $\T$. 
We will hence first construct a conservative extension $\rT$ of $\T$, which allow us to quantify over first-order formulas of the language of $\T$.
In order to achieve this we will take an approach that is inspired by Horsten's theory $\TC$~\cite{Horsten11}. There are however a few crucial differences between our approach and~\cite{Horsten11}, as follows. 
While $\TC$~\cite{Horsten11} is an extension of Peano Arithmetic $\PA$, our work can be used for an arbitrary theory $\T$. 
Further, while $\TC$~\cite{Horsten11}  relies on numbers to encode formulas, our approach uses multi-sorted logic and introduces additional sorts for formulas, terms, and variables, yielding a rather straightforward definition of  models and proof of consistency for the extended theory. 

The basic idea of our approach is to redefine the syntax and semantics of first-order logic in first-order logic itself. As such, there will be three levels of reasoning.
As usual, we have the meta level logic (the informal logical reasoning going on in the paper), and the object level (the formal logical reasoning we reason about on the meta level). In addition we will have a logic embedded in the terms of object level logic formulas. We will refer to this level as the \emph{reflective level}, and refer to formulas/functions/terms/expressions as reflective formulas/functions/terms/expressions. 

\subsection{Signature}
In a first step, we extend the signature $\Sigma$ and the set of sorts $\ssorts[\Sigma]$ of our base theory $\T$ with the vocabulary to be able to talk about variables, terms, and formulas.

\definitionReflectiveSignature

As this definition is rather lengthy we will break down the intended semantics of all newly introduced symbols. We can split the definitions into two parts: (i) one formalizing the syntax and (ii)  one formalizing the semantics of our reflective first-order logic.

\paragraph{(i) Reflective syntax.} Our reflective syntax is formalized as follows. 
\begin{description}
  \item[Variables]
    The sort $\rvar\sigma$ is used to represent the countably infinite set of variables $\ovar\sigma$. The two functions $\rvz\sigma$, and $\rvs\sigma$ that are added to the signature can be thought of as the constructors for this infinite set of variables. This means $\rvz\sigma$ is intended to be interpreted as the variable $\x 0$, $\rvs\sigma(\rvz\sigma)$ is meant to be interpreted as $\x 1$, and so on. We introduce the following syntactic sugar for variables:
    \begin{align*}
      \rvn\sigma{i + 1} &=  \rvs\sigma(\rvn\sigma i) & \text{for $i \geq 0$}
    \end{align*}

  \item[Terms] 
    We use the sort $\rterm\sigma$ to represent terms of sort $\oterm\sigma$. On the meta level terms are defined inductively, as follows.
    
    The base case is a variable. Since variables and terms are of different sorts, we need the function $\rtoterm\sigma$ to turn variables into terms. This function is intended to be interpreted as the identity function. 

    The step case of the inductive definition is building terms out of function symbols and other terms. Therefore, we need to introduce a reflective function symbol $\refl f$ for every function $f$ in the signature. The $\refl f$ is intended to be interpreted as the function symbol $f$, while $f$ itself is interpreted as an actual function.

  \item[Formulas] 
    As for terms, formulas $\oform$ are defined inductively on the meta level. 

    For atomic formulas we introduce a reflective equality symbol $\req\sigma$ for each sort $\sigma$ and a reflective version $\refl P$ for every predicate symbol $P$. 
    Even though it's not strictly necessary we introduce a nullary reflective connective $\rbot$ is intended to be interpreted as the formula $\obot$.

    Complex formulas are built from atomic formulas and connectives, or quantifiers. Therefore we introduce a functionally complete set of reflective connectives, namely $\ror$, and $\rnot$. As it will help in terms of readability, we will use infix notation for $\ror$, and drop the parenthesis for $\rnot$ if there is no ambiguity.

    In order to formalize quantification we introduce a function $\rforallS\sigma$ for each sort. We will write $\rforallA{x}\sigma p$ for the term $\rforallS\sigma(x, p)$.
\end{description}

\paragraph{(ii) Reflective semantics.}  For axiomatising the meaning of formulas, we will use syntactic representations of the semantic structures needed to define the semantics of first-order logic. 

\begin{description}
  \item[Environment] In order to define the meaning of a quantifier, we redefine the meaning of a variable within the scope of the quantifier. Therefore we will use a stack of variable interpretations, which we will call an environment. The idea is that a variable $\rvn \sigma i$ is freely interpreted in an empty environment $\rez$, while it is interpreted as $t$ if the tuple $\tup{\rvn\sigma i, t}$ was pushed on the stack using $\res\sigma(e, \rvn \sigma i, t)$. This setting becomes more clear in Sections \ref{sect:refl:axioms}-\ref{sect:refl:model}, where we axiomatise the meaning and define a model of the reflective theory. 

  \item[Evaluation] To make use of the environment, we need a reflective evaluation function for terms $\reval\sigma$ and $\revalv\sigma$ that corresponds to interpreting terms and variables in some model $\I$ of the reflective theory.

  \item[Satisfaction] Finally, we have our reflective satisfaction relation $\rmodels$. We write $\rmodelsANobrack{e}{p}$ for $\rmodels(e, p)$, which can roughly be interpreted as ``the interpretation $\I$ partially defined by $e$ satisfies $p$''. Our truth $\mtruth$ predicate in the Tarskian sense is $\mtruth(x) = \rmodelsA{\rez}{x}$.

\end{description}

\subsection{Axiomatisation}%
\label{sect:refl:axioms}
We now formalize our semantics. We relate reflective with non-reflective function and predicate symbols, by defining the meaning of the reflective satisfaction relation $\rmodels$, and the meaning of the reflective evaluation functions $\reval\sigma$, and $\revalv\sigma$. 
All axioms we list are implicitly universally quantified, and one instance of them will be present for every sort $\sigma, \tau \in \ssorts[\Sigma]$. Finally, the reflective extension $\rT$ of our base theory $\T$ is the union of all these axioms and $\T$.

\paragraph{Reflective variable interpretation.} 
As already mentioned, the interpretation of variables in an empty environment $\rez$ is undefined. In contrast an environment to which a variable $v$, and a value $x$ is pushed, evaluates the variable $v$ to $x$. Hence, 
\newcommand\axtag[1]{\tag{\ensuremath{\axtagfont{Ax}_{#1}}}}
  \begin{align}
  \axtag{\revalv{0}}\label{refl_ax:evalv0}\revalv{\sigma}(\res{\sigma }(e, v, x), v)  &= x\\
  \axtag{\revalv{1}}\label{refl_ax:evalv1}v \oneq v' \oimpl \revalv{\sigma}(\res{\sigma }(e, v', x), v) &= \revalv{\sigma}(e, v) \\
  \axtag{\revalv{2}}\label{refl_ax:evalv2}\revalv{\sigma}(\res{\tau}(e, w, x), v)  &= \revalv{\sigma}(e, v) && \text{ for $\sigma \neq \tau$ }
  \end{align}

\paragraph{Reflective evaluation.} 
The function symbol $\reval\sigma$ defines the value of a reflective term $t$, and thereby maps the reflective functions $\refl f$ to their non-reflective counter parts $f$. For variables $\reval\sigma$, the evaluation to $\revalv\sigma$ is used. 
  \begin{align*}
    \axtag{\reval{var}}\label{refl_ax:evalVar}\reval{\sigma}(e, \rtoterm{\sigma}(v)) &= \revalv{\sigma}(e, v)\\
    \axtag{\reval{f}}\label{refl_ax:evalTerm}\reval{\sigma}(e, \refl{f}(t_1, ..., t_n)) &= f(\reval{\sigma_1}(e, t_1), ..., \reval{\sigma_n}(e, t_n)) \\
                                                                                        &\text{\hspace{10mm}for $f: \sigma_1 \sprod ... \sprod \sigma_n \sarrow \sigma \in \Sigma$}
  \end{align*} 

\paragraph{Reflective satisfaction.} 
The predicate symbol $\rmodels\sigma$ defines the truth of a formula with respect to some variable interpretation. To this end, the meaning of the reflective connectives and the quantifiers  in terms is defined by their respective object-level counterparts, as follows: 
  \begin{align}
    \axtag{\req{}}\label{refl_ax:eq}\rmodelsA{e}{x \req{\sigma} y}    &\oiff \reval{\sigma}(e, x) \oeq  \reval{\sigma}(e, y)\\
    \axtag{P}     \label{refl_ax:pred}\rmodelsA{e}{\refl{P}(t_1, ..., t_n)}    &\oiff P(\reval{\sigma_1}(e, t_1), ..., \reval{\sigma_n}(e, t_n)) &&\text{for $P: \spred{\sigma_1 \sprod ... \sprod \sigma_n}$}\\
    \axtag{\rbot} \label{refl_ax:bot}\rmodelsA{e}{\rbot}                      &\oiff \obot\\
    \axtag{\rnot} \label{refl_ax:not}\rmodelsA{e}{\rnot \phi}                 &\oiff \onot\rmodelsA{e}{\phi}\\
    \axtag{\ror}  \label{refl_ax:or}\rmodelsA{e}{\phi \ror \psi}             &\oiff \rmodelsA{e}{\phi} \oor \rmodelsA{e}{\psi}\\
    \axtag{\rforallS{}}\label{refl_ax:forall}\rmodelsA{e}{\rforallA{v}{\sigma}{\phi}} &\oiff \oforall x: \sigma. \rmodelsA{\res{\sigma}(e,v,x)}{\phi}
  \end{align}

\subsection{Consistency and Conservativeness}

As we have now specified our theory, we next ensure that (i) $\rT$ is indeed a conservative extension of $\T$ and (ii)  $\rT$ is consistent. In general, we cannot ensure that $\rT$ is consistent, since already the base theory $\T$ could have been inconsistent. Hence we will show that $\rT$ is consistent if $\T$ is consistent.

In order to prove (i) and (ii), that is conservativeness and consistency of $\rT$,
we introduce the notion of a \emph{reflective model $\rM$}, that is based on a model $\M$ of $\T$.
The basic idea is that $\rM$ interprets every symbol in the base theory $\T$ as it would be interpreted in $\M$, hence every formula in $\Lang{\T}$ is true in $\rM$ iff it is true in $\M$. Due to soundness and completeness of first-order logic we get that $\rT$ is indeed a conservative extension of $\T$.
Further, due to the fact that for every model of $\M$ of $\T$ we have a model $\rM$ of $\rT$, we also have that $\rT$ is consistent if $\T$ is consistent.
In order to ensure this reasoning is correct we need to ensure that $\rM$ also satisfies the axioms we introduced for reflective theories. This will be done by interpreting the new reflective sort $\rform$ as the set of first order formulas $\oform$, and interpreting the sort $\rterm{\sigma}$ as terms of sort $\oterm{\sigma}$. 

\definitionReflectiveModel%
\label{sect:refl:model}

We now need to ensure that our reflective interpretation $\rM$ is indeed a model of $\rT$ if $\M$ is a model of $\T$.

\begin{theorem}[Reflective model]
  $$
    \M \vDash \T \miff \rM \vDash \rT 
  $$
\end{theorem}
\begin{proof}

The ``$\mlpmi$'' part of the biconditional is trivial since $\T \subset \rT$, and $\rM$ interprets all symbols of the original signature in the same way as $\M$.

For the same reason as before we have that $\M \vDash \T \mimpl \rM \vDash \T$. Hence we are left to show that $\rM \vdash \rT \setminus \T$. 
  This follows from  the axioms we introduced in \ref{sect:refl:axioms} in natural language, as well as from  the meta level semantics of first-order logic, by also making sure that our meta level and our object level definitions match.
\end{proof}

\subsection{Truth predicate}
We showed that our theory $\rT$ is indeed a conservative extension of $\T$. 
Next we prove that  $\rT$ behaves in the way we need it for axiomatising induction. That is, we need to make sure that $\rT$ has a truth predicate, allowing us to quantify over formulas and thus defining the induction scheme as a single formula.

As in \cite{Horsten11},  we use a G\"odel encoding to state that our theory $\rT$ has a truth predicate.  Usually, a G\"odel encoding maps variables, terms, and formulas to numerals. Since our theory $\rT$ does not necessarily contain number symbols, we need to use a more general notion of a G\"odel encoding, namely that it maps variables, terms, and formulas in our base language $\Lang{\T}$ to terms in our extended language $\Lang{\rT}$. That is, we map formulas $\oform$ to terms of sort $\rform$, variables $\ovar\sigma$ to $\rvar\sigma$ and $\oterm\sigma$ to $\rterm\sigma$. Formally, we define our G\"odel encoding as follows:

\definitionGoedelEncoding

With our G\"odel encoding at hand,  we can now show that $\rT$ contains a truth predicate $\mtruth[\phi]$ for $\T$, namely the formula $\rmodelsANobrack{\rez}{\goedel{\phi}}$. To this end, we have the following result. 

\begin{theorem}[ Truth Predicate ]
  \label{theorem:truthpredicate}
  \begin{align*}
    \mforall \phi \in \Lang{\T}.\Big( \rT \vDash \phi \oiff \rmodelsA{\rez}{\goedel{\phi}} \Big)
  \end{align*} \qed
\end{theorem}


\IfNotExtended{%
In order to prove Theorem~\ref{theorem:truthpredicate}, we strengthen its statements,
such that it holds not only for the empty reflective interpretation $\rez$,
but also for every reflective interpretation built from terms $\rez$, and $\lambda e. \res\sigma(e, \rvn \sigma i, \x i^\sigma)$.
This is necessary, as (i) our axiomatisation of $\rforallS\sigma$ defines the meaning of a quantifier in terms of a different reflective environment
and (ii) we relate finitary representation of variables (based on $\rvz\sigma$ and $\lambda x. \rvs\sigma(x)$) to the infinite set of variables used for standard first-order logic syntax. %
For the detailed proof of Theorem~\ref{theorem:truthpredicate}, we refer to the extended version of this paper \todo{cite}. 
}
\IfExtended{
  \input{proof_truth_predicate.tex}
}

%% file: proof_truth_predicate.tex
  \begin{proof}

  In order to proof this theorem inductively we will need to strengthen our goal to:
  
  \begin{align*}
    \mforall e \in \semanticStack. \rT \vdash \phi \oiff \rmodelsA{e}{\goedel{\phi}}
  \end{align*} 

  where we define the set $\semanticStack$ inductively as the least set such that 
  \begin{itemize}
    \item $\rez \in \semanticStack$
    \item $e \in \semanticStack \mand \sigma \in \ssorts \mand i \in \N \mimpl \res\sigma(e, \rvn \sigma i, \x i) \in \semanticStack$
  \end{itemize}

  Next we will rewrite our goal to to:

  \begin{align*}
    \mforall e \in \semanticStack. \rT \vdash \phi \miff \rT \vdash \rmodelsA{e}{\goedel{\phi}}
  \end{align*} 

  We will now prove the theorem by induction on the structure of $\phi$. Most cases can be proven by simply unfolding of definitions of the Gödel encoding, applying the axioms of $\rT$, and applying the induction hypothesis.

    \newcommand\tp[1]{\rmodelsA{e}{#1}}
    \renewcommand\entry[2]{\miff & \rT \vdash #1 & \text{by #2}\\}
  \begin{description}

  \item[case $\alpha \oor \beta$.]
    \begin{align*}
      \rT \vdash \tp{\goedel{\alpha \oor \beta}} 
      \entry{  \tp{( \goedel{\alpha} \ror \goedel{\beta} )}  }{\eqref{goedel_def:or}}
      \entry{  \tp{\goedel{\alpha}} \lor \tp{\goedel{\beta}}   }{\eqref{refl_ax:or}}
      \entry{ \alpha \lor \tp{\goedel{\beta}}  }{I.H.}
      \entry{ \alpha \lor \beta  }{I.H.}
                                                   & \square
    \end{align*} 

  \item[case $\onot \psi$.]
    \begin{align*}
      \rT \vdash \tp{\goedel{\onot\psi}} 
      \entry{  \tp{ \rnot\goedel{\psi} }  }{\eqref{goedel_def:not}}
      \entry{  \onot \tp{\goedel{\psi}}  }{\eqref{refl_ax:not}}
      \entry{ \onot \psi  }{I.H.}
                                                   & \square
    \end{align*} 

  \item[case $\obot$.]
    \begin{align*}
      \rT \vdash \tp{\goedel{\obot}} 
      \entry{  \tp{ \rbot }  }{\eqref{goedel_def:bot}}
      \entry{  \obot  }{\eqref{refl_ax:bot}}
                & \square
    \end{align*} 

  \item[case $\oforall \x i:\sigma.\phi$.]
    \begin{align*}
      \rT \vdash \tp{\goedel{\oforall \x i : \sigma. \phi}} 
      \entry{ \tp{\rforallA{\rvn{\sigma}i}{\sigma}{\goedel{\phi} }}}{\eqref{goedel_def:forall}}
      \entry{ \oforall \x i . \rmodelsA{\res{\sigma}(e, \rvn{\sigma}i,\x i)}{\goedel{\phi}}}{\eqref{refl_ax:forall}}
      \entry{ \oforall \x i . \phi}{I.H.}
                & \square
    \end{align*} 
    \end{description}

    The more involved cases are the ones dealing with atomic formulas. The reason why this cannot be dealt with by simple unfolding of definitions, is that we cannot ensure that every object level variable $\rvn \sigma i$ is interpreted as the variable $\x i$ of sort $\sigma$, since this would have required adding an infinite number of axioms to our base theory $\T$.

    \begin{description}
  \item[case $P(t_1,...,t_n)$.]
    By the existence of a sound and complete proof system for first-order logic, we can rewrite our induction hypothesis as 
    $$ \rT \vDash P(t_1, ..., t_n) \miff \rT \vDash \tp{\goedel{P(t_1...t_n)}} $$
    which is equivalent to the statement
    $$ 
    \mforall \M \vDash \rT. \Big(\M \vDash P(t_1, ..., t_n)\Big)
      \miff 
    \mforall \M \vDash \rT. \Big(\M \vDash \tp{\goedel{P(t_1...t_n)}}\Big)
    $$
    which can again be rewritten to 
    $$ 
    \mexists \M \vDash \rT. \Big(\M \not\vDash P(t_1, ..., t_n)\Big)
      \miff 
    \mexists \M \vDash \rT. \Big(\M \not\vDash \tp{\goedel{P(t_1...t_n)}}\Big)
    $$

    We will proof both directions of the biconditional separately:
    \begin{description}
      \item[case ``$\mlpmi$''] 
      In order to show that this implication holds we will show, that if there is a model $\tup{\dty,\I}$ such that $\tup{\dty,\I}  \not\vDash \tp{\goedel{P(t_1,...,t_n)}}$, then there is another model $\tup{\dty,\hat\I}$ such that $\tup{\dty,\hat\I}\not\vDash P(t_1,...,t_n)$. 

      The idea is that $\hat\I$ differs from $\I$ only in the interpretation of the variables. In exact the variables in $\hat\I$ are interpreted in the same way as reflective variables $\rvn\sigma i$ are interpreted in $\I$. Therefore the interpretation of the evaluation of a term $\goedel{t}$ in $\hat\I$ will be the same as the interpretation of $t$ in $\I$, hence $\tup{\dty,\hat\I}$ will satisfy $ \tp{\goedel{\phi}} $ iff $\tup{\dty,\I}$ satisfies $\phi$, which implies what we want to show.

        More formally: 


        Let $\tup{D,I}$ be a model of $\rT$.
        We define $\hat\I$ as follows.

        \begin{align*}
          \hat\I(x) &=\begin{cases} 
            \I(\revalv{\sigma}(\rez,\rvn\sigma i )) & \text{if } x = \x i \mand \x i \in \ovar\sigma \\
            \I(x) & \text{otherwise}
          \end{cases}
        \end{align*}

        \begin{proposition}
          \label{claim:x_norm} 
          \begin{align*}
            \I(\reval\sigma(e,\goedel{t})) = \hat\I(t) 
          \end{align*}
        \end{proposition}
        \begin{proof}
        We apply induction on $t$.
        \begin{description}
          \item[case $f(t_1, ..., t_n)$.] 
            \begin{align*}
                & \hat\I(f(t_1,...,t_n)) \\
                &= \hat\I(f)(\hat\I(t_1), ..., \hat\I(t_n)) \\
                &= \I(f)(\hat\I(t_1), ..., \hat\I(t_n)) &\text{by definition of $\hat\I$} \\
                &= \I(f)(\I(\reval\sigma(e,\goedel{t_1})), ..., \I(\reval\sigma(e,\goedel{t_n}))) &\text{by I.H.} \\
                &= \I(\reval\sigma(e,\goedel{f(t_1, ..., t_n)})) &\text{by \eqref{goedel_def:fun} and \eqref{refl_ax:evalTerm}} \\
                 &\square
             \end{align*}

          \item[case $\x i \in \ovar\sigma$.] 
            \begin{align*}
                    &  \hat\I(\x i) \\
                    &= \I(\revalv{\sigma}(e,\rvn\sigma i )) & \text{by definition of $\hat\I$}\\
                    &= \I(\reval{\sigma}(e,\rtoterm\sigma( \rvn\sigma i ) )) & \text{by \eqref{refl_ax:evalVar}}\\
                                 &=\I(\reval\sigma(e,\goedel{\x i})) & \text{by \eqref{goedel_def:var}}\\
                             &\square
            \end{align*}
        \end{description}
        \end{proof}

        Now that we have showed that Proposition \ref{claim:x_norm} holds we can reason as follows:

        \begin{align*}
               &\hat\I \vDash P(t_1, ..., t_n)\\
          \iff &\hat\I(P) \ni \tup{\hat\I(t_1), ..., \hat\I(t_n)}\\
          \iff &\hat\I(P) \ni \tup{\I(\reval\sigma(e, \goedel{t_1})), ..., \I(\reval\sigma(e, \goedel{t_n}))} & \text{ by Proposition \ref{claim:x_norm} }\\
          \iff &    \I(P) \ni \tup{\I(\reval\sigma(e, \goedel{t_1})), ..., \I(\reval\sigma(e, \goedel{t_n}))} & \text{ by definition of $\hat\I$ }\\
          \iff &    \I \vDash P(\reval\sigma(e, \goedel{t_1}), ..., \reval\sigma(e, \goedel{t_n})) \\
          \iff &    \I \vDash \tp{\refl P(\goedel{t_1}, ..., \goedel{t_n})} & \text{by \eqref{refl_ax:pred}} \\
          \iff &    \I \vDash \tp{\goedel{P(t_1, ..., t_n)}} & \text{by \eqref{goedel_def:pred}} \\
        \end{align*}

        From this we can conclude that if there is a model $\tup{\dty,\I}$ such that $\tup{\dty,\I} \not\vDash \tp{\goedel{P(t_1,...,t_n)}}$, then the model $\tup{\dty,\hat\I} \not\vDash P(t_1,...,t_n)$. $\square$

      \item[case $``\mimpl''$]
        The idea for this case similar to the idea for the case before: We assume there is a model $\tup{\dty,\I}$ that makes $P(t_1,...,t_n)$ false, and from that build another model $\tup{\dty,\hat\I}$ that makes the $\tp{\goedel{P(t_1, ...,t_n)}}$ false. In this case our new model will differ from the old one not in the interpretation of the variables $\x i$, but in the interpretation of the evaluation of the reflective variables $\rvn\sigma i$, in such a way that the evaluation of $\rvn\sigma i$ in $\hat\I$ will always be interpreted as the same value as the interpretation of $\x i$ in $\I$.


        Let $\tup{D,I}$ be a model of $\rT$. 
        We define the interpretation $\hat\I$ as follows.

        \begin{align*}
          \hat\I(x) &= \I(x) & \text{ for $x \neq \revalv\sigma$ }\\
          \hat\I(\revalv{\sigma})(e, v) &= \begin{cases} 
            \I(\x i) & \text{if $e = \rez \mand v = \rvn\sigma i$}\\
            \hat\I(t) & \text{if $e = \res\sigma(e', v, t)$}\\
            \hat\I(\revalv{\sigma})(e', v) & \text{if $e = \res\sigma(e', v', t) \mand v \neq v'$}\\
            \hat\I(\revalv{\sigma})(e', v) & \text{if $e = \res\tau(e', v', t) \mand \tau \neq \sigma$}\\
          \end{cases}
        \end{align*}
        Note that the definition of $\hat\I(\revalv\sigma)$, is not a partial definition, since we defined $e \in \semanticStack$ inductively.

        \begin{proposition}
          \label{claim:v_norm}
          \begin{align}
            \I(t) = \hat\I(\reval\sigma(e,\goedel{t}))
          \end{align}
        \end{proposition}

        \begin{proof}
          We apply induction on $t$:
          \begin{description}
            \item[case $f(t_1, ..., t_n)$.] 
              \begin{align*}
                \I(f(t_1,...,t_n)) &= \I(f)(\I(t_1), ..., \I(t_n)) \\
                &= \hat\I(f)(\I(t_1), ..., \I(t_n)) &\text{by definition of $\hat\I$} \\
                &= \hat\I(f)(\I(\reval\sigma(e,\goedel{t_1})), ..., \I(\reval\sigma(e,\goedel{t_n}))) &\text{by I.H.} \\
                &= \hat\I(f(\reval\sigma(e,\goedel{t_1}), ..., \reval\sigma(e,\goedel{t_n}))) \\
                &= \hat\I(\reval\sigma(e,\refl{f}(\goedel{t_1}, ..., \goedel{t_n}))) &\text{by \eqref{refl_ax:evalTerm}} \\
                &= \hat\I(\reval\sigma(e,\goedel{f(t_1, ..., t_n)})) &\text{by \eqref{goedel_def:fun}} \\
                 &\square
              \end{align*}

            \item[case $\x i \in \ovar\sigma$.] 
              Induction on $e$:
              \begin{description}
                \item[case $\rez$]
                  \begin{align*}
                    \I(\x i) &= \hat\I(\revalv\sigma)(\rez, \rvn\sigma i) & \text{by definition of $\hat\I$}\\
                             &= \hat\I(\revalv\sigma)(\rez, \goedel{\x i}) & \text{by \eqref{goedel_def:var}}\\
                             &= \hat\I(\revalv\sigma(\rez, \goedel{\x i})) \\
                             &\square
                  \end{align*}

                \item[case $\res\tau(e', \rvn\tau j, \x j)$] \ 

                  \begin{description}

                    \item[case $\tau \neq \sigma$] 
                      \begin{align*}
                        & \hat\I(\revalv\sigma(\res\tau(e', \rvn\tau j, \x j), \goedel{\x i})) \\
                        &= \hat\I(\revalv\sigma)(\res\tau(e', \rvn\tau j, \x j), \goedel{\x i})\\
                        &= \hat\I(\revalv\sigma)(e', \goedel{\x i}) &\text{by definition of $\hat\I$ }\\
                        &= \hat\I(\revalv\sigma(e', \goedel{\x i}))\\
                        &= \I(\x i) & \text{by I.H.}\\
                                 &\square
                      \end{align*}

                    \item[case $\tau = \sigma \mand i \neq j$] 
                      \begin{align*}
                        &\hat\I(\revalv\sigma(\res\tau(e', \rvn\tau j, \x j), \goedel{\x i})) \\
                        &= \hat\I(\revalv\sigma(\res\sigma(e',\rvn\sigma j, \x j), \goedel{\x i})) \\
                        &= \hat\I(\revalv\sigma)(\res\sigma(e', \rvn\sigma j, \x j), \goedel{\x i})\\
                        &= \hat\I(\revalv\sigma)(e', \goedel{\x i}) &\text{by definition of $\hat\I$ }\\
                        &= \hat\I(\revalv\sigma(e', \goedel{\x i}))\\
                        &= \I(\x i) & \text{by I.H.}\\
                                 &\square
                      \end{align*}

                    \item[case $\tau = \sigma \mand i = j$] 
                      \begin{align*}
                        &\hat\I(\revalv\sigma(\res\tau(e', \rvn\tau j, \x j), \goedel{\x i})) \\
                        &= \hat\I(\revalv\sigma(\res\sigma(e', \rvn\sigma i, \x i), \goedel{\x i})) \\
                        &= \hat\I(\revalv\sigma)(\res\sigma(e', \rvn\sigma i, \x i), \goedel{\x i}) \\
                        &= \hat\I(\x i) & \text{by definition of $\hat\I$} \\
                        &= \I(\x i) & \text{by definition of $\hat\I$} \\
                        &\square
                      \end{align*}

                  \end{description}
              \end{description}
          \end{description}
          This concludes the end of the proof of Proposition~\ref{claim:v_norm}. 
          \end{proof}

          Therefore we can reason analogous to before.

        \begin{align*}
                &\I \not\vDash P(t_1, ..., t_n) \\
         \miff  &\I(P) \not\ni P(t_1, ..., t_n) \\
         \miff  &    \I(P) \not\ni \tup{\hat\I(\reval\sigma(e,\goedel{t_1})), ..., \hat\I(\reval\sigma(e,\goedel{t_n}))} & \text{by Proposition \ref{claim:v_norm}} \\
         \miff  &\hat\I(P) \not\ni \tup{\hat\I(\reval\sigma(e,\goedel{t_1})), ..., \hat\I(\reval\sigma(e,\goedel{t_n}))} & \text{by definition of $\hat\I$}\\
         \miff  &\hat\I    \not\vDash P(\reval\sigma(e,\goedel{t_1}), ...,\reval\sigma(e,\goedel{t_n})) \\
         \miff  &\hat\I    \not\vDash \tp{\refl P(\goedel{t_1}, ...,\goedel{t_n})} & \text{by \eqref{refl_ax:pred}}\\
         \miff  &\hat\I    \not\vDash \tp{\goedel{ P(t_1, ...,t_n) }} & \text{by \eqref{goedel_def:pred}}\\
        \end{align*}

        Now that we have established this know that if there is a model $\tup{\domain,\I}$ such that $\tup{\domain,\I} \not\vDash P(t_1,...,t_n)$, then there is a model $\tup{\domain,\hat\I}$ such that $\hat\I \not\vDash \tp{\goedel{ P(t_1, ...,t_n) }}$. 
        This concludes our proof of the case ``$\mimpl$'' of the biconditional, and therefore as well the induction case for $\phi = P(t_1,...,t_n)$ in the proof of our main \autoref{theorem:truthpredicate}.
                $\square$

    \end{description}
  \item[case $s \oeq t$.] This case is exactly the same as for uninterpreted predicates, since equality can be thought of as a binary predicate.

  \end{description}

  We have therefore established that \autoref{theorem:truthpredicate} indeed holds.
  \end{proof}

%% file: induction_by_reflection.tex
We next show how to build a finite theory that entails the first-order
induction scheme $\foInd{ \tau }$. 
For the sake of simplicity we will first have a look at the rather familiar case of Peano Arithmetic, and present a generalisation of the same approach in the following subsection.


\subsection{Natural Numbers}
\label{sect:finitePA}
\newcommand\rQ{\ensuremath{\refl\Q}}
In order to finitely axiomatise \PA, we need a finite fragment of \PA{} to start with. The obvious choice is $\Q$, which we define as $\PA$ without the induction formulas. 
We then  build the reflective extension $\rQ$, with the following  two
essential properties. First,  it has a sort $\rform$ of formulas, hence
we can quantify over this sort.
Second,  we have a truth predicate $\rQ$ for $\Q$, which means we can represent an arbitrary formula of $\Q$ in a single term in $\rQ$.

\newcommand\PAp{\reflInd{\Q{}}\xspace}
Now we can define a conservative extension  $\PAp$ of \PA. Therefore
we will add the following axiom to $\rQ$; we call this axiom as the
\emph{reflective induction axiom}: 

\newcommand\appl[2]{\msyn{True}[ #1, #2]}
\newcommand\applDef[2]{\rmodelsA{\res{\nat}(\rez, \rvz\nat, #2)}{#1}}
\begin{align*}
  \tag{\ensuremath{\refl{\mathInd}}}
  \label{ax:refl_ind_math}
  \oforall \phi: \rform.\Big(
          & \appl{\phi}{\zero} \oand  \\
          & \oforall n: \nat. (\appl{\phi}{n} \oimpl \appl{\phi}{\succ{n}}  )\\
          & \oimpl \oforall n: \nat. \appl{\phi}{n}
          \Big)  
\end{align*}
where 
$ \appl{\phi}{n} := \applDef{\phi}{n} $. Thus, we define $\PAp$ as 
\begin{align*}
  \PAp = \rQ \cup \s{ \dot\mathInd }
\end{align*}

\begin{theorem}
  $\PAp$ is a conservative extension of $\PA$
\end{theorem}

\begin{proof}
We will need the following auxiliary formula:
\begin{align}
  \label{lemma:substitution} 
  \rT \vDash \rmodelsA{\res\sigma(e, \goedel{\x i}, t)}{\goedel{\phi[\x i]}} \oiff  \rmodelsA{e}{\goedel{\phi[t]}}
\end{align}
This formula  holds by induction over $\phi$. 
Proving that \PAp is a conservative extension of \PA reduces showing
that every formula in $\Lang{\PA}$ is provable in $\PAp$ iff it is
provable in $\PA$.
We next prove both directions of this property. 

\newcommand\IndInst{\mathbf I}
\paragraph{(1) $\forall \phi \in \Lang{\PA}.(\PA \vDash \phi \mimpl \PAp
  \vDash \phi)$} To this end, we show that all axioms of \PA are
derivable in \PAp. Since \Q is a subset of both \PA and \PAp,
we only need to deal with the induction axioms.
Let $\phi[0] \oand \oforall n. (\phi[n] \oimpl \phi[n+1]) \oimpl
\forall n.\phi[n]$ be an arbitrary instance of the first-order
mathematical induction scheme $\foInd\nat$.
Let us instantiate the reflective induction
axiom~\eqref{ax:refl_ind_math} with $\goedel{\phi[\x0]}$. We obtain 
\begin{align*}
           \PAp \vdash 
           &\appl{\goedel{\phi[\x0]}}{\zero} \oand   \\
           &\oforall n: \nat. (\appl{\goedel{ \phi[\x0] }}{n} \oimpl \appl{\goedel{ \phi[\x0] }}{\succ{n}}  ) \\
           &\oimpl \oforall n: \nat. \appl{\goedel{ \phi[\x0] }}{n}
\end{align*}

which expands to 

\begin{align*}
           \PAp \vdash 
           &\applDef{\goedel{\phi[\x0]}}{\zero} \oand  \\
           &\oforall n: \nat. (\applDef{\goedel{ \phi[\x0] }}{n} \oimpl \applDef{\goedel{ \phi[\x0] }}{\succ{n}}  )\\
     \oimpl& \oforall n: \nat. \applDef{\goedel{ \phi[\x0] }}{n}
\end{align*}

By formula~\eqref{lemma:substitution}, we can derive 

\renewcommand\appl[2]{\rmodelsA{\rez}{\goedel{#1[#2]}}}
\begin{align*}
    \PAp \vdash 
           &\appl{\phi}{\zero} \oand  \\
           &\oforall n: \nat. (\appl{\phi}{n} \oimpl \appl{\phi}{\succ{n}}  )\\
     \oimpl& \oforall n: \nat. \appl{\phi}{n}
\end{align*}
Applying \autoref{theorem:truthpredicate}, the fact that $\lambda
x.\rmodelsA{\rez}{x}$ is our truth predicate, we finally get
\renewcommand\appl[2]{#1[#2]}
\begin{align*}
    \PAp \vdash 
    \appl{\phi}{\zero} \oand  \oforall n: \nat. (\appl{\phi}{n} \oimpl \appl{\phi}{\succ{n}}  )
     \oimpl \oforall n: \nat. \appl{\phi}{n}
\end{align*}

\paragraph{(2) $\oforall \phi. (\PAp \vDash \phi \mimpl \PA \vDash \phi)$} 
We prove by contraposition.
Suppose we have some formula $\phi$ such that $\PA \not\vDash
\phi$. Hence there is a counter-model $\M$, such that $\M \vDash \PA$
but $\M \not\vDash \phi$. Since $\Q \subset \PA$, it holds that $\M
\vDash \Q$. Thanks to  Section~\ref{sect:refl:model} we can extend the
model $\M$ to the reflective model $\rM$ such that $\rM \vDash \rQ$,
and that $\rM \not\vDash \phi$. We are thus left with establish that
$\rM$ is a model of $\PAp$. 

In $\rM$ the sort $\rform$ is interpreted as the actual set of formulas $\oform$. Therefore let $\phi[\x 0]$ be an arbitrary of these formulas. Since $\M \vDash \PA$, we have that $\rM \vDash \PA$, which implies that
\renewcommand\appl[2]{#1[#2]}
\begin{align*}
    \rM \vDash 
    \appl{\phi}{\zero} \oand  \oforall n: \nat. (\appl{\phi}{n} \oimpl \appl{\phi}{\succ{n}}  )
     \oimpl \oforall n: \nat. \appl{\phi}{n}
\end{align*}
By \autoref{theorem:truthpredicate}, we get  
\renewcommand\appl[2]{\rmodelsA{\rez}{\goedel{#1[#2]}}}
\begin{align*}
    \rM \vDash 
           &\appl{\phi}{\zero} \oand  \\
           &\oforall n: \nat. (\appl{\phi}{n} \oimpl \appl{\phi}{\succ{n}}  )\\
     \oimpl& \oforall n: \nat. \appl{\phi}{n}
\end{align*}
which, using formula~\ref{lemma:substitution}, can  be rewritten to 
\renewcommand\appl[2]{\msyn{True}[ #1, #2]}
\begin{align*}
  \rM \vDash 
           &\appl{       {\phi[x_0]}}{\zero} \oand  \\
           &\oforall n: \nat. (\appl{       \phi[x_0]}{n} \oimpl \appl{       \phi[x_0]}{\succ{n}}  )\\
     \oimpl& \oforall n: \nat. \appl{       \phi[x_0]}{n}
\end{align*}

Since $\rM$ interprets $\rform$ formulas exactly as the set $\oform$
and $\phi[x_0]$ is an arbitrary formula,
we conclude that the  reflective induction axiom~\ref{ax:refl_ind}
holds for $\rM$.
Therefore, $\rM$ models $\PAp$ but not $\phi$, which concludes our
proof.
\end{proof}


\subsection{Arbitrary datatypes}
The result of Section~\ref{sect:finitePA}  can be lifted to arbitrary
datatypes.
Therefore,  we translate the meta-level definition of the induction
scheme $\foInd\tau$ for datatypes $\dty_\tau$  to an equivalent
reflective version.
That is,  for a theory $\T$, we build $\T'$ by adding the axiom
$\refl{\foInd\tau}$ to $\rT$ for every datatype $\dty_\tau$ in the
theory, as follows: 

\renewcommand\appl[2]{\msyn{True}[ #1, #2]}
\renewcommand\applDef[2]{\rmodelsA{\res{\tau}(\rez, \rvz\tau, #2)}{#1}}

\begin{align*}
  \label{ax:refl_ind}
  \tag{\ensuremath{\refl{\foInd\tau}}}
  \oforall \phi: \rform.\Big(
    \bigwedge_{c\in\ctors} case_{\phi,c} \oimpl \forall x:\tau. \appl{\phi}{x}
          \Big)  
\end{align*}
where 
\begin{align*}
  case_{\phi,c} &:= \universalClosure[ x_1, ..., x_n ] \Big( \bigwedge_{i \in recursive_c} \appl{\phi}{\x i} \oimpl \appl{\phi}{c( x_1, ..., x_n )} \Big)\\
  recursive_c &:=  \s{ i \mid \sdom[\Sigma](c, i) = \tau } \\
  \appl{\phi}{n} &:= \applDef{\phi}{n}
\end{align*}

\newcommand\disj[1]{\ensuremath{\mathsf{Disj_{#1}}}\xspace}
\newcommand\inj[1]{\ensuremath{\mathsf{Inj_{#1}}}\xspace}

In the case of extending \Q{} to a conservative extension of \PA, the
axioms of constructor disjointness \disj{\nat}, and injectivity
\inj{\nat} were already present in \Q.
Thus, for an arbitrary inductive theory $\T$ with inductive datatypes
$\dty_\T$, we define the reflective inductive extension $\reflInd{\T}$ as follows:
\begin{align*}
  \label{def:refl_ind_theory} \reflInd{\T} = \T &\cup \s{ \eqref{ax:refl_ind}, \disj{\tau},  \inj{\tau} \mid \dty_\tau \in \dty_\T } 
\end{align*}

%% file: experiments.tex
\input{generated/mscBenchmarkDefs.tex}

\input{experiments_figures.tex}

\newcommand\True[1]{\rmodelsA{\rez}{#1}}
\newcommand\problem[1]{\ensuremath{\mathbf{#1}}\xspace}
\newcommand\pRefl{\problem{Refl}}
\newcommand\pInd{\problem{Ind}}

\newcommand\pReflEasy{\problem{Refl_0}}
\newcommand\pReflHard{\problem{Refl_1}}

In order to evaluate the practical viability of the techniques introduced Sections~\ref{sect:reflExt}-\ref{sect:finiteInd}, we performed two set of experiments, denoted as $\pRefl$ and $\pInd$ and described next. 

\paragraph{Setup}
Note that our work introduces many new function symbols, and axioms which might blow up the proof search space, even if induction is not involved at all. Therefore, in our first experiment  $\pRefl$ we wanted to evaluate the feasibility of reasoning in the reflective extension of a theory. 

\IfExtended{
Since many of our benchmarks have the same axioms, but different conclusions a list of all the base theories is given in table~\ref{tab:theories}.

\begin{table}
  \center
  \tabAllTheories%
  \caption{Theories used for the experiments.}%
  \label{tab:theories}
\end{table}

}

\pRefl itself consists of two groups of benchmarks \pReflEasy{}, and \pReflHard{}. \pRefl{} is the simplest one. For every theory $\T$ in some set of base theories, and every axiom $\alpha \in \T$ we try to proof the validity of $\refl\T \vdash \True{\goedel{\alpha}}$. Since we established that $\lambda x. \True{x}$ is the truth predicate of $\T$, and the fact that $\alpha$ is an axiom, we know that these consequence assertions indeed hold.
\pReflHard{} involves reasoning in the reflective extension $\rT$ of some theory as well. But in this case not the reflective version of the axioms, but the reflective versions of some simple consequence of $\T$ are to be proven. 
\IfExtended{
Table~\ref{tab:reflConj}, lists all conjectures and the related theories, that are to be proven in this set of benchmarks.

\begin{table}
  \center
  \scalebox{0.9}{%
    \tabReflConj{}
  }
  \caption{Conjectures and theories used for the benchmark set \pReflHard{}. The decimal numbers used are abbreviations for the corresponding numerals.}%
  \label{tab:reflConj}
\end{table}
}

The benchmarks in the second experiment \pInd{} are a set of crafted properties that require inductive reasoning. Every problem $\T \vDash \phi$ in this set of benchmarks is addressed in two ways. Firstly, proving it directly for the solvers that support induction natively, and secondly, translating the problem to $\reflInd{\T} \vDash \phi$. %
\IfNotExtended{%
For a description of the exact axioms and conjectures used in each of our benchmarks we refer to the extended version of this paper \todo{cite}.
}
\IfExtended{%
Table~\ref{tab:indConj} lists the base theories and conjectures used for this experiment.

\begin{table}
  \center
  \scalebox{0.8}{%
    \tabIndConj{}
  }
  \caption{Conjectures and theories used for the benchmark set \pInd{}}%
  \label{tab:indConj}
\end{table}
}

All benchmarks, as well as a program for generating reflective, and reflective inductive extensions of theories, and G\"odel encodings for conjectures can be found at \github\footnote{\url{https://github.com/joe-hauns/msc-automating-induction-via-reflection}}. As the different solvers we used for evaluation support different input formats, our tool supports serializing problems into these various formats.

We used two (non-disjoint) sets of solvers. Firstly, solvers that support induction natively, and secondly various general-purpose theorem provers that are able to deal with multi-sorted quantified first-order logic, hence induction using the reflective extension. 

The solvers considered where the SMT-solvers \cvc4 and \z3, the superposition-based first-order theorem prover \vampire, the higher-order theorem prover \zipperposition that uses a combination of superposition and term rewriting, and the inductive theorem prover \zeno, that is designed to proof inductive properties of a \haskell-like programming language. Since \vampire{} in many cases uses incomplete strategy, per default it was run with a complete strategy forced as well. This configuration if referred to as \solver{VampireComplete}. 
\zipperposition{} supports replacing equalities by dedicated rewrite rules, which comes at the cost of the theoretical loss of some provable problems, but yields a significant gain of performance in practice. \zipperposition{} with these rewrite rules enabled will be referred to as \solver{ZipRewrite}. 
\cvc4 allows for theory exploration which was shown to be helpful for inductive reasoning in \cite{InductionWithGeneralization20}. \cvc4 with this heuristic enabled is referred to as \solver{Cvc4Gen}. 

We ran each solver with a timeout of 10 seconds per problem.

\begin{table}
  \center
  \tabRefl%
  \caption{%
    Results of the experiment \pRefl.
  }%
  \label{tab:results-refl-results}
\end{table}

\paragraph{Results}
In the first part of Table~\ref{tab:results-refl-results} we can see the results of solvers proving reflective versions of axioms. What is striking is that the SMT solvers \cvc4, and \z3, can solver all benchmarks of this category, while the problem seems to be harder for the saturation based theorem provers. Further \ZipRewrite does pretty well in this class of benchmarks as well.  A potential reason for this difference in performance between the ordinary saturation approach and \ZipRewrite might have to do with the following: For \ZipRewrite equalities for function definitions of the reflective extensions are translated to rewrite rules that are oriented in way that they would intuitively be oriented by a human, this means that for example the axiom \eqref{refl_ax:evalTerm} can be evaluated as one would intuitively do. In contrast \vampire, using superposition with the Knuth-Bendix simplification ordering will orient this equality in the wrong way, which means that it won't be able to evaluate it in the intuitive way, which might be the reason for the difference in performance.

The second part of the table shows that the performance of the SMT-solvers drops as soon as more complex reasoning is involved. Especially the problems with conjectures involving existential quantification\footnote{These problem ids contain the substring ``exists'' in their id.} are hardly solved by the SMT solvers. This is not surprising since SMT solvers target at solving quantifier-free fragments of first-order logic.

\begin{table}%
  \center
  \scalebox{0.9}{%
    \tabIndResults
  }
  \caption{%
    Lists the results of running solvers on the benchmark set \pInd. For every solver $\solver{Slvr}$ that supports full first-order logic with equality as input, there is a solver $\reflInd{\solver{Slvr}}$ using the reflective inductive theory as an input instead of using the solvers native handling of induction. The greyed out cells mean that the problem cannot be translated to the solvers input format.
  }%
  \label{tab:results-refl-ind}
\end{table}

Table~\ref{tab:results-refl-ind} lists the results of the final experiment $\pInd$. As the first experiments have shown reasoning in the reflective theories is hard even for very simple conjectures, it is not surprising that it is even harder for problems that require inductive reasoning to solve. 
Nevertheless there are some problems that can be solved using the reflective inductive extension instead of built-in induction heuristics. The most striking result is that \z3 is able to solve benchmarks that involve induction, even though it is a SMT-solver without any support for inductive reasoning.

%% file: generated/mscBenchmarkDefs.tex
\newcommand\reflAxTheories{\ensuremath{\theoryFmt{N} + \theoryFmt{Leq} + \theoryFmt{Add} + \theoryFmt{Mul}}, \ensuremath{\theoryFmt{N} + \theoryFmt{L} + \theoryFmt{Pref} + \theoryFmt{App}}}
\newcommand\tabAllTheories{\begin{tabular}{|c|c|}\hline 
Name & Theory\\\hline

\hline \theoryFmt{N} & {$ \begin{matrix}
 \data \osyn{nat} = \osyn{zero} \mid \osyn{s}(\osyn{nat})\\

 \begin{matrix}  \end{matrix}\\
  \end{matrix} $} \\ 
\hline \theoryFmt{Leq} & {$ \begin{matrix}
  \begin{matrix} 
	\leq\softype{}\spred{\osyn{nat}\sprod{}\osyn{nat}} \end{matrix}\\
 	\oforall x. (x \leq x) & (1)\\
	\oforall x, y. ((x \leq y) \oimpl (x \leq \osyn{s}(y))) & (2)\\
 \end{matrix} $} \\ 
\hline \theoryFmt{Add} & {$ \begin{matrix}
  \begin{matrix} 
	+\softype{}\osyn{nat}\sprod{}\osyn{nat}\sarrow{}\osyn{nat} \end{matrix}\\
 	\oforall y. (\osyn{zero} + y) \oeq y & (1)\\
	\oforall x, y. (\osyn{s}(x) + y) \oeq \osyn{s}((x + y)) & (2)\\
 \end{matrix} $} \\ 
\hline \theoryFmt{Mul} & {$ \begin{matrix}
  \begin{matrix} 
	*\softype{}\osyn{nat}\sprod{}\osyn{nat}\sarrow{}\osyn{nat} \end{matrix}\\
 	\oforall y. (\osyn{zero} * y) \oeq \osyn{zero} & (1)\\
	\oforall x, y. (\osyn{s}(x) * y) \oeq (y + (x * y)) & (2)\\
 \end{matrix} $} \\ 
\hline \theoryFmt{L} & {$ \begin{matrix}
 \data \osyn{lst} = \osyn{nil} \mid \osyn{cons}(\osyn{nat}, \osyn{lst})\\

 \begin{matrix}  \end{matrix}\\
  \end{matrix} $} \\ 
\hline \theoryFmt{Pref} & {$ \begin{matrix}
  \begin{matrix} 
	\osyn{pref}\softype{}\spred{\osyn{lst}\sprod{}\osyn{lst}} \end{matrix}\\
 	\oforall x. \osyn{pref}(\osyn{nil},x) & (1)\\
	\oforall a, x. \onot{}\osyn{pref}(\osyn{cons}(a,x),\osyn{nil}) & (2)\\
	\oforall a, b, x, y. (\osyn{pref}(\osyn{cons}(a,x),\osyn{cons}(b,y)) \oiff (a \oeq b \oand \osyn{pref}(x,y))) & (3)\\
 \end{matrix} $} \\ 
\hline \theoryFmt{App} & {$ \begin{matrix}
  \begin{matrix} 
	+\hspace{-2mm}+ \softype{}\osyn{lst}\sprod{}\osyn{lst}\sarrow{}\osyn{lst} \end{matrix}\\
 	\oforall r. (\osyn{nil} +\hspace{-2mm}+  r) \oeq r & (1)\\
	\oforall a, l, r. (\osyn{cons}(a,l) +\hspace{-2mm}+  r) \oeq \osyn{cons}(a,(l +\hspace{-2mm}+  r)) & (2)\\
 \end{matrix} $} \\ 
\hline \theoryFmt{E} & {$ \begin{matrix}
  \begin{matrix} 
	\osyn{a}\softype{}\alpha & \osyn{b}\softype{}\alpha\\
	\osyn{c}\softype{}\alpha & \osyn{p}\softype{}\spred{\alpha}\\
	\osyn{q}\softype{}\spred{\alpha} & \osyn{r}\softype{}\spred{\alpha} \end{matrix}\\
  \end{matrix} $} \\ 
\hline \theoryFmt{Id} & {$ \begin{matrix}
  \begin{matrix} 
	\osyn{id}\softype{}\osyn{nat}\sarrow{}\osyn{nat} \end{matrix}\\
 	\oforall x. \osyn{id}(x) \oeq x & (1)\\
 \end{matrix} $} \\ 
\hline \theoryFmt{Eq} & {$ \begin{matrix}
  \begin{matrix} 
	\osyn{equal}\softype{}\spred{\osyn{nat}\sprod{}\osyn{nat}\sprod{}\osyn{nat}} \end{matrix}\\
 	\osyn{equal}(\osyn{zero},\osyn{zero},\osyn{zero}) \oiff \otop & (1)\\
	\oforall y, z. (\osyn{equal}(\osyn{zero},\osyn{s}(y),z) \oiff \obot) & (2)\\
	\oforall y, z. (\osyn{equal}(\osyn{zero},y,\osyn{s}(z)) \oiff \obot) & (3)\\
	\oforall x, z. (\osyn{equal}(\osyn{s}(x),\osyn{zero},z) \oiff \obot) & (4)\\
	\oforall x, y. (\osyn{equal}(\osyn{s}(x),y,\osyn{zero}) \oiff \obot) & (5)\\
	\oforall x, y, z. (\osyn{equal}(\osyn{s}(x),\osyn{s}(y),\osyn{s}(z)) \oiff \osyn{equal}(x,y,z)) & (6)\\
 \end{matrix} $} \\ 
\hline \theoryFmt{Rev} & {$ \begin{matrix}
  \begin{matrix} 
	\osyn{rev}\softype{}\osyn{lst}\sarrow{}\osyn{lst} \end{matrix}\\
 	\osyn{rev}(\osyn{nil}) \oeq \osyn{nil} & (1)\\
	\oforall x, xs. \osyn{rev}(\osyn{cons}(x,xs)) \oeq (\osyn{rev}(xs) +\hspace{-2mm}+  \osyn{cons}(x,\osyn{nil})) & (2)\\
 \end{matrix} $} \\ 
\hline \theoryFmt{Rev'} & {$ \begin{matrix}
  \begin{matrix} 
	\osyn{rev'}\softype{}\osyn{lst}\sarrow{}\osyn{lst} & \osyn{revAcc}\softype{}\osyn{lst}\sprod{}\osyn{lst}\sarrow{}\osyn{lst} \end{matrix}\\
 	\oforall x. \osyn{rev'}(x) \oeq \osyn{revAcc}(x,\osyn{nil}) & (1)\\
	\oforall acc. \osyn{revAcc}(\osyn{nil},acc) \oeq acc & (2)\\
	\oforall acc, x, xs. \osyn{revAcc}(\osyn{cons}(x,xs),acc) \oeq \osyn{revAcc}(xs,\osyn{cons}(x,acc)) & (3)\\
 \end{matrix} $} \\ 
\hline 
\end{tabular}
}

\newcommand\tabReflConj{\begin{tabular}{|c|c|c|}\hline 
Theory & Conjecture & id\\\hline

\hline \theoryFmt{E} & {$ \goedel{\oforall x: \alpha. x \oeq x} $} & \benchIdFmt{eqRefl} \\ 
\hline \theoryFmt{E} & {$ \goedel{\oforall x, y, z: \alpha. ((x \oeq y \oand y \oeq z) \oimpl x \oeq z)} $} & \benchIdFmt{eqTrans} \\ 
\hline \theoryFmt{E} & {$ \goedel{\osyn{p}(\osyn{a}) \oor \onot{}\osyn{p}(\osyn{a})} $} & \benchIdFmt{excludedMiddle-0} \\ 
\hline \theoryFmt{E} & {$ \goedel{\oforall x. (\osyn{p}(x) \oor \onot{}\osyn{p}(x))} $} & \benchIdFmt{excludedMiddle-1} \\ 
\hline \theoryFmt{E} & {$ \goedel{\oforall x. \osyn{p}(x) \oimpl \osyn{p}(\osyn{a})} $} & \benchIdFmt{universalInstance} \\ 
\hline \theoryFmt{E} & {$ \goedel{(\osyn{p}(\osyn{a}) \oimpl \osyn{q}(\osyn{b})) \oiff (\onot{}\osyn{q}(\osyn{b}) \oimpl \onot{}\osyn{p}(\osyn{a}))} $} & \benchIdFmt{contraposition-0} \\ 
\hline \theoryFmt{E} & {$ \goedel{\oforall x, y. ((\osyn{p}(x) \oimpl \osyn{q}(y)) \oiff (\onot{}\osyn{q}(y) \oimpl \onot{}\osyn{p}(x)))} $} & \benchIdFmt{contraposition-1} \\ 
\hline \theoryFmt{E} & {$ \goedel{((\osyn{p}(\osyn{a}) \oand \osyn{q}(\osyn{b})) \oimpl \osyn{r}(\osyn{c})) \oiff (\osyn{p}(\osyn{a}) \oimpl (\osyn{q}(\osyn{b}) \oimpl \osyn{r}(\osyn{c})))} $} & \benchIdFmt{currying-0} \\ 
\hline \theoryFmt{E} & {$ \goedel{\oforall x, y, z. (((\osyn{p}(x) \oand \osyn{q}(y)) \oimpl \osyn{r}(z)) \oiff (\osyn{p}(x) \oimpl (\osyn{q}(y) \oimpl \osyn{r}(z))))} $} & \benchIdFmt{currying-1} \\ 
\hline \ensuremath{\theoryFmt{N} + \theoryFmt{Add}} & {$ \goedel{(\osyn{1} + \osyn{2}) \oeq \osyn{3}} $} & \benchIdFmt{addGround-0} \\ 
\hline \ensuremath{\theoryFmt{N} + \theoryFmt{Add}} & {$ \goedel{(\osyn{8} + \osyn{5}) \oeq \osyn{13}} $} & \benchIdFmt{addGround-1} \\ 
\hline \ensuremath{\theoryFmt{N} + \theoryFmt{Add}} & {$ \goedel{\oexists x. (\osyn{8} + x) \oeq \osyn{13}} $} & \benchIdFmt{addExists} \\ 
\hline \ensuremath{\theoryFmt{N} + \theoryFmt{Add}} & {$ \goedel{\oexists z. \oforall x. (z + x) \oeq x} $} & \benchIdFmt{existsZeroAdd} \\ 
\hline \ensuremath{\theoryFmt{N} + \theoryFmt{Add} + \theoryFmt{Mul}} & {$ \goedel{(\osyn{3} * \osyn{4}) \oeq \osyn{12}} $} & \benchIdFmt{mulGround} \\ 
\hline \ensuremath{\theoryFmt{N} + \theoryFmt{Add} + \theoryFmt{Mul}} & {$ \goedel{\oexists x. (\osyn{3} * x) \oeq \osyn{12}} $} & \benchIdFmt{mulExists} \\ 
\hline \ensuremath{\theoryFmt{N} + \theoryFmt{Add} + \theoryFmt{Mul}} & {$ \goedel{\oexists z. \oforall x. (z * x) \oeq z} $} & \benchIdFmt{existsZeroMul} \\ 
\hline \ensuremath{\theoryFmt{N} + \theoryFmt{L} + \theoryFmt{App}} & {$ \goedel{(\osyn{nil} +\hspace{-2mm}+  \osyn{cons}(\osyn{7},\osyn{nil})) \oeq \osyn{cons}(\osyn{7},\osyn{nil})} $} & \benchIdFmt{appendGround-0} \\ 
\hline \ensuremath{\theoryFmt{N} + \theoryFmt{L} + \theoryFmt{App}} & {$ \goedel{(\osyn{cons}(\osyn{3},\osyn{nil}) +\hspace{-2mm}+  \osyn{cons}(\osyn{7},\osyn{nil})) \oeq \osyn{cons}(\osyn{3},\osyn{cons}(\osyn{7},\osyn{nil}))} $} & \benchIdFmt{appendGround-1} \\ 
\hline \ensuremath{\theoryFmt{N} + \theoryFmt{L} + \theoryFmt{App}} & {$ \goedel{\oexists x. (\osyn{cons}(\osyn{3},\osyn{nil}) +\hspace{-2mm}+  x) \oeq \osyn{cons}(\osyn{3},\osyn{cons}(\osyn{7},\osyn{nil}))} $} & \benchIdFmt{appendExists} \\ 
\hline \ensuremath{\theoryFmt{N} + \theoryFmt{L} + \theoryFmt{App}} & {$ \goedel{\oexists n. (n +\hspace{-2mm}+  \osyn{cons}(\osyn{7},\osyn{nil})) \oeq \osyn{cons}(\osyn{7},\osyn{nil})} $} & \benchIdFmt{existsNil} \\ 
\hline 
\end{tabular}
}
\newcommand\tabIndConj{\begin{tabular}{|c|c|c|}\hline 
Theory & Conjecture & id\\\hline

\hline \ensuremath{\theoryFmt{N} + \theoryFmt{Add}} & {$ \oforall x, y. (x + y) \oeq (y + x) $} & \benchIdFmt{addCommut} \\ 
\hline \ensuremath{\theoryFmt{N} + \theoryFmt{Add} + \theoryFmt{Mul}} & {$ \oforall x, y. (x * y) \oeq (y * x) $} & \benchIdFmt{mulCommut} \\ 
\hline \ensuremath{\theoryFmt{N} + \theoryFmt{Add}} & {$ \oforall x, y, z. (x + (y + z)) \oeq ((x + y) + z) $} & \benchIdFmt{addAssoc} \\ 
\hline \ensuremath{\theoryFmt{N} + \theoryFmt{Add} + \theoryFmt{Mul}} & {$ \oforall x, y, z. (x * (y * z)) \oeq ((x * y) * z) $} & \benchIdFmt{mulAssoc} \\ 
\hline \ensuremath{\theoryFmt{N} + \theoryFmt{Add}} & {$ \oforall x. (x + \osyn{zero}) \oeq x $} & \benchIdFmt{addNeutral} \\ 
\hline \ensuremath{\theoryFmt{N} + \theoryFmt{Add} + \theoryFmt{Mul}} & {$ \oforall x. (x * \osyn{1}) \oeq x $} & \benchIdFmt{addNeutral-0} \\ 
\hline \ensuremath{\theoryFmt{N} + \theoryFmt{Add} + \theoryFmt{Mul}} & {$ \oforall x. (\osyn{1} * x) \oeq x $} & \benchIdFmt{addNeutral-1} \\ 
\hline \ensuremath{\theoryFmt{N} + \theoryFmt{Add} + \theoryFmt{Mul}} & {$ \oforall x. (x * \osyn{zero}) \oeq \osyn{zero} $} & \benchIdFmt{mulZero} \\ 
\hline \ensuremath{\theoryFmt{N} + \theoryFmt{Add} + \theoryFmt{Mul}} & {$ \oforall x, y, z. (x * (y + z)) \oeq ((x * y) + (x * z)) $} & \benchIdFmt{distr-0} \\ 
\hline \ensuremath{\theoryFmt{N} + \theoryFmt{Add} + \theoryFmt{Mul}} & {$ \oforall x, y, z. ((y + z) * x) \oeq ((y * x) + (z * x)) $} & \benchIdFmt{distr-1} \\ 
\hline \ensuremath{\theoryFmt{N} + \theoryFmt{Leq}} & {$ \oforall x, y, z. (((x \leq y) \oand (y \leq z)) \oimpl (x \leq z)) $} & \benchIdFmt{leqTrans} \\ 
\hline \ensuremath{\theoryFmt{N} + \theoryFmt{Leq}} & {$ \oforall x. (\osyn{zero} \leq x) $} & \benchIdFmt{zeroMin} \\ 
\hline \ensuremath{\theoryFmt{N} + \theoryFmt{Leq} + \theoryFmt{Add}} & {$ \oforall x, y. (x \leq (x + y)) $} & \benchIdFmt{addMonoton-0} \\ 
\hline \ensuremath{\theoryFmt{N} + \theoryFmt{Leq} + \theoryFmt{Add}} & {$ \oforall x. (x \leq (x + x)) $} & \benchIdFmt{addMonoton-1} \\ 
\hline \ensuremath{\theoryFmt{N} + \theoryFmt{Add} + \theoryFmt{Id}} & {$ \oforall x, y. (\osyn{id}(x) + y) \oeq (y + x) $} & \benchIdFmt{addCommutId} \\ 
\hline \ensuremath{\theoryFmt{N} + \theoryFmt{L} + \theoryFmt{App}} & {$ \oforall x, y, z. (x +\hspace{-2mm}+  (y +\hspace{-2mm}+  z)) \oeq ((x +\hspace{-2mm}+  y) +\hspace{-2mm}+  z) $} & \benchIdFmt{appendAssoc} \\ 
\hline \ensuremath{\theoryFmt{N} + \theoryFmt{L} + \theoryFmt{Pref} + \theoryFmt{App}} & {$ \oforall x, y. \osyn{pref}(x,(x +\hspace{-2mm}+  y)) $} & \benchIdFmt{appendMonoton} \\ 
\hline \ensuremath{\theoryFmt{N} + \theoryFmt{Eq}} & {$ \oforall x. \osyn{equal}(x,x,x) $} & \benchIdFmt{allEqRefl} \\ 
\hline \ensuremath{\theoryFmt{N} + \theoryFmt{Eq}} & {$ \oforall x, y, z. (\osyn{equal}(x,y,z) \oiff (x \oeq y \oand y \oeq z)) $} & \benchIdFmt{allEqDefsEquality} \\ 
\hline \ensuremath{\theoryFmt{N} + \theoryFmt{L} + \theoryFmt{App} + \theoryFmt{Rev}} & {$ \oforall x. \osyn{rev}(\osyn{rev}(x)) \oeq x $} & \benchIdFmt{revSelfInvers} \\ 
\hline \ensuremath{\theoryFmt{N} + \theoryFmt{L} + \theoryFmt{App} + \theoryFmt{Rev}} & {$ \oforall x. (x +\hspace{-2mm}+  (\osyn{rev}(x) +\hspace{-2mm}+  x)) \oeq ((x +\hspace{-2mm}+  \osyn{rev}(x)) +\hspace{-2mm}+  x) $} & \benchIdFmt{revAppend-0} \\ 
\hline \ensuremath{\theoryFmt{N} + \theoryFmt{L} + \theoryFmt{App} + \theoryFmt{Rev}} & {$ \oforall x. \osyn{rev}((x +\hspace{-2mm}+  (x +\hspace{-2mm}+  x))) \oeq \osyn{rev}(((x +\hspace{-2mm}+  x) +\hspace{-2mm}+  x)) $} & \benchIdFmt{revAppend-1} \\ 
\hline \ensuremath{\theoryFmt{N} + \theoryFmt{L} + \theoryFmt{App} + \theoryFmt{Rev} + \theoryFmt{Rev'}} & {$ \oforall x. \osyn{rev}(x) \oeq \osyn{rev'}(x) $} & \benchIdFmt{revsEqual} \\ 
\hline 
\end{tabular}
}

\newcommand\tabSolvers{\begin{tabular}{|c|c|c|c|}\hline 
Solver & Induction & Input format & Commandline Options\\\hline

\hline \text{\solver{Cvc4}} & \yes  & \Smtlib  & \code{ --quant-ind } \\ 
\hline \text{\solver{Cvc4Gen}} & \yes  & \Smtlib  & \code{ --conjecture-gen --quant-ind } \\ 
\hline \text{\solver{Z3}} & \no  & \Smtlib  & default mode \\ 
\hline \text{\solver{Vampire}} & \yes  & \Smtlib  & \code{ --schedule casc --induction struct } \\ 
\hline \text{\solver{VampireComplete}} & \yes  & \Smtlib  & \code{ --induction struct -s 1 } \\ 
\hline \text{\solver{Zipperposition}} & \yes  & \Zf  & default mode \\ 
\hline \text{\solver{ZipRewrite}} & \yes  & \Zf  & default mode \\ 
\hline \text{\solver{Zeno}} & \yes  & \ZenoHaskell  & default mode \\ 
\hline 
\end{tabular}
}


\newcommand\tabReflAxResults{\begin{tabular}{|c|c|c|c|c|c|c|c|} 
\resultsTableBenchmarkColHeader  & \resultTableSolverCell{\text{\solver{Cvc4}}} & \resultTableSolverCell{\text{\solver{Cvc4Gen}}} & \resultTableSolverCell{\text{\solver{Z3}}} & \resultTableSolverCell{\text{\solver{Vampire}}} & \resultTableSolverCell{\text{\solver{VampireComplete}}} & \resultTableSolverCell{\text{\solver{Zipperposition}}} & \resultTableSolverCell{\text{\solver{ZipRewrite}}}\\\hline

\hline \code{N+Leq+Add+Mul-ax0} & \resultTabYes  & \resultTabYes  & \resultTabYes  & \resultTabYes  & \resultTabYes  & \resultTabYes  & \resultTabYes  \\ 
\hline \code{N+Leq+Add+Mul-ax1} & \resultTabYes  & \resultTabYes  & \resultTabYes  & \resultTabNo   & \resultTabNo   & \resultTabNo   & \resultTabYes  \\ 
\hline \code{N+Leq+Add+Mul-ax2} & \resultTabYes  & \resultTabYes  & \resultTabYes  & \resultTabYes  & \resultTabYes  & \resultTabNo   & \resultTabYes  \\ 
\hline \code{N+Leq+Add+Mul-ax3} & \resultTabYes  & \resultTabYes  & \resultTabYes  & \resultTabNo   & \resultTabNo   & \resultTabNo   & \resultTabYes  \\ 
\hline \code{N+Leq+Add+Mul-ax4} & \resultTabYes  & \resultTabYes  & \resultTabYes  & \resultTabYes  & \resultTabYes  & \resultTabNo   & \resultTabYes  \\ 
\hline \code{N+Leq+Add+Mul-ax5} & \resultTabYes  & \resultTabYes  & \resultTabYes  & \resultTabNo   & \resultTabNo   & \resultTabNo   & \resultTabYes  \\ 
\hline \code{N+L+Pref+App-ax0} & \resultTabYes  & \resultTabYes  & \resultTabYes  & \resultTabYes  & \resultTabYes  & \resultTabYes  & \resultTabYes  \\ 
\hline \code{N+L+Pref+App-ax1} & \resultTabYes  & \resultTabYes  & \resultTabYes  & \resultTabYes  & \resultTabYes  & \resultTabNo   & \resultTabYes  \\ 
\hline \code{N+L+Pref+App-ax2} & \resultTabYes  & \resultTabYes  & \resultTabYes  & \resultTabNo   & \resultTabNo   & \resultTabNo   & \resultTabNo   \\ 
\hline \code{N+L+Pref+App-ax3} & \resultTabYes  & \resultTabYes  & \resultTabYes  & \resultTabYes  & \resultTabYes  & \resultTabNo   & \resultTabYes  \\ 
\hline \code{N+L+Pref+App-ax4} & \resultTabYes  & \resultTabYes  & \resultTabYes  & \resultTabNo   & \resultTabNo   & \resultTabNo   & \resultTabYes  \\ 
\hline 
\end{tabular}
}
\newcommand\tabReflConjResults{\begin{tabular}{|c|c|c|c|c|c|c|c|} 
\resultsTableBenchmarkColHeader  & \resultTableSolverCell{\text{\solver{Cvc4}}} & \resultTableSolverCell{\text{\solver{Cvc4Gen}}} & \resultTableSolverCell{\text{\solver{Z3}}} & \resultTableSolverCell{\text{\solver{Vampire}}} & \resultTableSolverCell{\text{\solver{VampireComplete}}} & \resultTableSolverCell{\text{\solver{Zipperposition}}} & \resultTableSolverCell{\text{\solver{ZipRewrite}}}\\\hline

\hline \code{eqRefl} & \resultTabYes  & \resultTabYes  & \resultTabYes  & \resultTabYes  & \resultTabYes  & \resultTabYes  & \resultTabYes  \\ 
\hline \code{eqTrans} & \resultTabYes  & \resultTabYes  & \resultTabYes  & \resultTabNo   & \resultTabNo   & \resultTabNo   & \resultTabYes  \\ 
\hline \code{excludedMiddle-0} & \resultTabYes  & \resultTabYes  & \resultTabYes  & \resultTabYes  & \resultTabYes  & \resultTabYes  & \resultTabYes  \\ 
\hline \code{excludedMiddle-1} & \resultTabYes  & \resultTabYes  & \resultTabYes  & \resultTabYes  & \resultTabYes  & \resultTabYes  & \resultTabYes  \\ 
\hline \code{universalInstance} & \resultTabNo   & \resultTabNo   & \resultTabYes  & \resultTabYes  & \resultTabYes  & \resultTabYes  & \resultTabYes  \\ 
\hline \code{contraposition-0} & \resultTabYes  & \resultTabYes  & \resultTabYes  & \resultTabYes  & \resultTabYes  & \resultTabYes  & \resultTabYes  \\ 
\hline \code{contraposition-1} & \resultTabYes  & \resultTabYes  & \resultTabYes  & \resultTabNo   & \resultTabNo   & \resultTabNo   & \resultTabYes  \\ 
\hline \code{currying-0} & \resultTabYes  & \resultTabYes  & \resultTabYes  & \resultTabYes  & \resultTabYes  & \resultTabNo   & \resultTabYes  \\ 
\hline \code{currying-1} & \resultTabYes  & \resultTabYes  & \resultTabYes  & \resultTabNo   & \resultTabNo   & \resultTabNo   & \resultTabYes  \\ 
\hline \code{addGround-0} & \resultTabYes  & \resultTabYes  & \resultTabYes  & \resultTabYes  & \resultTabYes  & \resultTabYes  & \resultTabYes  \\ 
\hline \code{addGround-1} & \resultTabYes  & \resultTabYes  & \resultTabNo   & \resultTabNo   & \resultTabNo   & \resultTabYes  & \resultTabYes  \\ 
\hline \code{addExists} & \resultTabNo   & \resultTabNo   & \resultTabNo   & \resultTabNo   & \resultTabNo   & \resultTabNo   & \resultTabYes  \\ 
\hline \code{existsZeroAdd} & \resultTabNo   & \resultTabNo   & \resultTabNo   & \resultTabNo   & \resultTabNo   & \resultTabNo   & \resultTabNo   \\ 
\hline \code{mulGround} & \resultTabYes  & \resultTabYes  & \resultTabYes  & \resultTabNo   & \resultTabNo   & \resultTabNo   & \resultTabYes  \\ 
\hline \code{mulExists} & \resultTabNo   & \resultTabNo   & \resultTabNo   & \resultTabNo   & \resultTabNo   & \resultTabNo   & \resultTabYes  \\ 
\hline \code{existsZeroMul} & \resultTabNo   & \resultTabNo   & \resultTabNo   & \resultTabNo   & \resultTabNo   & \resultTabNo   & \resultTabNo   \\ 
\hline \code{appendGround-0} & \resultTabYes  & \resultTabYes  & \resultTabYes  & \resultTabYes  & \resultTabYes  & \resultTabYes  & \resultTabYes  \\ 
\hline \code{appendGround-1} & \resultTabYes  & \resultTabYes  & \resultTabYes  & \resultTabNo   & \resultTabNo   & \resultTabYes  & \resultTabYes  \\ 
\hline \code{appendExists} & \resultTabNo   & \resultTabNo   & \resultTabNo   & \resultTabNo   & \resultTabNo   & \resultTabNo   & \resultTabYes  \\ 
\hline \code{existsNil} & \resultTabNo   & \resultTabNo   & \resultTabYes  & \resultTabYes  & \resultTabYes  & \resultTabNo   & \resultTabYes  \\ 
\hline 
\end{tabular}
}
\newcommand\tabIndResults{\begin{tabular}{|c|c|c|c|c|c|c|c|c|c|c|c|c|c|c|} 
\resultsTableBenchmarkColHeader  & \resultTableSolverCell{\text{\solver{Cvc4}}} & \resultTableSolverCell{\text{\solver{Cvc4Gen}}} & \resultTableSolverCell{\text{\solver{Vampire}}} & \resultTableSolverCell{\text{\solver{VampireComplete}}} & \resultTableSolverCell{\text{\solver{Zipperposition}}} & \resultTableSolverCell{\text{\solver{ZipRewrite}}} & \resultTableSolverCell{\text{\solver{Zeno}}} & \resultTableSolverCell{\ensuremath{\reflIndSolver{\text{\solver{Cvc4}}}}} & \resultTableSolverCell{\ensuremath{\reflIndSolver{\text{\solver{Cvc4Gen}}}}} & \resultTableSolverCell{\ensuremath{\reflIndSolver{\text{\solver{Z3}}}}} & \resultTableSolverCell{\ensuremath{\reflIndSolver{\text{\solver{Vampire}}}}} & \resultTableSolverCell{\ensuremath{\reflIndSolver{\text{\solver{VampireComplete}}}}} & \resultTableSolverCell{\ensuremath{\reflIndSolver{\text{\solver{Zipperposition}}}}} & \resultTableSolverCell{\ensuremath{\reflIndSolver{\text{\solver{ZipRewrite}}}}}\\\hline

\hline \code{addCommut} & \resultTabYes  & \resultTabYes  & \resultTabYes  & \resultTabYes  & \resultTabYes  & \resultTabYes  & \resultTabYes  & \resultTabNo   & \resultTabNo   & \resultTabNo   & \resultTabNo   & \resultTabNo   & \resultTabNo   & \resultTabNo   \\ 
\hline \code{mulCommut} & \resultTabNo   & \resultTabNo   & \resultTabNo   & \resultTabNo   & \resultTabNo   & \resultTabNo   & \resultTabNo   & \resultTabNo   & \resultTabNo   & \resultTabNo   & \resultTabNo   & \resultTabNo   & \resultTabNo   & \resultTabNo   \\ 
\hline \code{addAssoc} & \resultTabYes  & \resultTabYes  & \resultTabYes  & \resultTabYes  & \resultTabYes  & \resultTabYes  & \resultTabYes  & \resultTabNo   & \resultTabNo   & \resultTabNo   & \resultTabNo   & \resultTabNo   & \resultTabNo   & \resultTabNo   \\ 
\hline \code{mulAssoc} & \resultTabNo   & \resultTabNo   & \resultTabNo   & \resultTabNo   & \resultTabNo   & \resultTabNo   & \resultTabNo   & \resultTabNo   & \resultTabNo   & \resultTabNo   & \resultTabNo   & \resultTabNo   & \resultTabNo   & \resultTabNo   \\ 
\hline \code{addNeutral} & \resultTabYes  & \resultTabYes  & \resultTabYes  & \resultTabYes  & \resultTabYes  & \resultTabYes  & \resultTabYes  & \resultTabNo   & \resultTabNo   & \resultTabNo   & \resultTabNo   & \resultTabNo   & \resultTabNo   & \resultTabNo   \\ 
\hline \code{addNeutral-0} & \resultTabYes  & \resultTabYes  & \resultTabYes  & \resultTabYes  & \resultTabYes  & \resultTabYes  & \resultTabYes  & \resultTabNo   & \resultTabNo   & \resultTabNo   & \resultTabNo   & \resultTabNo   & \resultTabNo   & \resultTabNo   \\ 
\hline \code{addNeutral-1} & \resultTabYes  & \resultTabYes  & \resultTabYes  & \resultTabYes  & \resultTabYes  & \resultTabYes  & \resultTabYes  & \resultTabNo   & \resultTabNo   & \resultTabNo   & \resultTabNo   & \resultTabNo   & \resultTabNo   & \resultTabNo   \\ 
\hline \code{mulZero} & \resultTabYes  & \resultTabYes  & \resultTabYes  & \resultTabYes  & \resultTabYes  & \resultTabYes  & \resultTabYes  & \resultTabNo   & \resultTabNo   & \resultTabNo   & \resultTabNo   & \resultTabNo   & \resultTabNo   & \resultTabYes  \\ 
\hline \code{distr-0} & \resultTabNo   & \resultTabNo   & \resultTabNo   & \resultTabNo   & \resultTabNo   & \resultTabNo   & \resultTabNo   & \resultTabNo   & \resultTabNo   & \resultTabNo   & \resultTabNo   & \resultTabNo   & \resultTabNo   & \resultTabNo   \\ 
\hline \code{distr-1} & \resultTabNo   & \resultTabNo   & \resultTabNo   & \resultTabNo   & \resultTabNo   & \resultTabYes  & \resultTabNo   & \resultTabNo   & \resultTabNo   & \resultTabNo   & \resultTabNo   & \resultTabNo   & \resultTabNo   & \resultTabNo   \\ 
\hline \code{leqTrans} & \resultTabNo   & \resultTabNo   & \resultTabNo   & \resultTabNo   & \resultTabNo   & \resultTabNo   & \resultTabNA   & \resultTabNo   & \resultTabNo   & \resultTabNo   & \resultTabNo   & \resultTabNo   & \resultTabNo   & \resultTabNo   \\ 
\hline \code{zeroMin} & \resultTabYes  & \resultTabYes  & \resultTabYes  & \resultTabYes  & \resultTabYes  & \resultTabYes  & \resultTabNA   & \resultTabNo   & \resultTabNo   & \resultTabYes  & \resultTabYes  & \resultTabNo   & \resultTabNo   & \resultTabYes  \\ 
\hline \code{addMonoton-0} & \resultTabNo   & \resultTabNo   & \resultTabNo   & \resultTabNo   & \resultTabNo   & \resultTabNo   & \resultTabNA   & \resultTabNo   & \resultTabNo   & \resultTabNo   & \resultTabNo   & \resultTabNo   & \resultTabNo   & \resultTabNo   \\ 
\hline \code{addMonoton-1} & \resultTabNo   & \resultTabNo   & \resultTabNo   & \resultTabNo   & \resultTabNo   & \resultTabNo   & \resultTabNA   & \resultTabNo   & \resultTabNo   & \resultTabNo   & \resultTabNo   & \resultTabNo   & \resultTabNo   & \resultTabNo   \\ 
\hline \code{addCommutId} & \resultTabNo   & \resultTabYes  & \resultTabYes  & \resultTabYes  & \resultTabYes  & \resultTabYes  & \resultTabYes  & \resultTabNo   & \resultTabNo   & \resultTabNo   & \resultTabNo   & \resultTabNo   & \resultTabNo   & \resultTabNo   \\ 
\hline \code{appendAssoc} & \resultTabYes  & \resultTabYes  & \resultTabYes  & \resultTabYes  & \resultTabYes  & \resultTabYes  & \resultTabYes  & \resultTabNo   & \resultTabNo   & \resultTabNo   & \resultTabNo   & \resultTabNo   & \resultTabNo   & \resultTabNo   \\ 
\hline \code{appendMonoton} & \resultTabYes  & \resultTabYes  & \resultTabYes  & \resultTabYes  & \resultTabYes  & \resultTabYes  & \resultTabYes  & \resultTabNo   & \resultTabNo   & \resultTabNo   & \resultTabNo   & \resultTabNo   & \resultTabNo   & \resultTabNo   \\ 
\hline \code{allEqRefl} & \resultTabYes  & \resultTabYes  & \resultTabYes  & \resultTabYes  & \resultTabYes  & \resultTabYes  & \resultTabNo   & \resultTabNo   & \resultTabNo   & \resultTabYes  & \resultTabNo   & \resultTabNo   & \resultTabNo   & \resultTabNo   \\ 
\hline \code{allEqDefsEquality} & \resultTabYes  & \resultTabYes  & \resultTabNo   & \resultTabNo   & \resultTabYes  & \resultTabYes  & \resultTabNo   & \resultTabNo   & \resultTabNo   & \resultTabNo   & \resultTabNo   & \resultTabNo   & \resultTabNo   & \resultTabNo   \\ 
\hline \code{revSelfInvers} & \resultTabNo   & \resultTabNo   & \resultTabNo   & \resultTabNo   & \resultTabYes  & \resultTabNo   & \resultTabNo   & \resultTabNo   & \resultTabNo   & \resultTabNo   & \resultTabNo   & \resultTabNo   & \resultTabNo   & \resultTabNo   \\ 
\hline \code{revAppend-0} & \resultTabNo   & \resultTabNo   & \resultTabNo   & \resultTabNo   & \resultTabNo   & \resultTabYes  & \resultTabNo   & \resultTabNo   & \resultTabNo   & \resultTabNo   & \resultTabNo   & \resultTabNo   & \resultTabNo   & \resultTabNo   \\ 
\hline \code{revAppend-1} & \resultTabNo   & \resultTabNo   & \resultTabNo   & \resultTabNo   & \resultTabNo   & \resultTabYes  & \resultTabNo   & \resultTabNo   & \resultTabNo   & \resultTabNo   & \resultTabNo   & \resultTabNo   & \resultTabNo   & \resultTabNo   \\ 
\hline \code{revsEqual} & \resultTabNo   & \resultTabNo   & \resultTabNo   & \resultTabNo   & \resultTabNo   & \resultTabNo   & \resultTabNo   & \resultTabNo   & \resultTabNo   & \resultTabNo   & \resultTabNo   & \resultTabNo   & \resultTabNo   & \resultTabNo   \\ 
\hline 
\end{tabular}
}

%% file: experiments_figures.tex
\newcommand\tabRefl{
  \begin{tabular}{|c|c|c|c|c|c|c|c|} 
  \resultsTableBenchmarkColHeader  & \resultTableSolverCell{\text{\solver{Cvc4}}} & \resultTableSolverCell{\text{\solver{Cvc4Gen}}} & \resultTableSolverCell{\text{\solver{Z3}}} & \resultTableSolverCell{\text{\solver{Vampire}}} & \resultTableSolverCell{\text{\solver{VampireComplete}}} & \resultTableSolverCell{\text{\solver{Zipperposition}}} & \resultTableSolverCell{\text{\solver{ZipRewrite}}}\\\hline

  \multicolumn{8}{|c|}{\pReflEasy}\\
  \hline \code{N+Leq+Add+Mul-ax0} & \resultTabYes  & \resultTabYes  & \resultTabYes  & \resultTabYes  & \resultTabYes  & \resultTabYes  & \resultTabYes  \\ 
  \hline \code{N+Leq+Add+Mul-ax1} & \resultTabYes  & \resultTabYes  & \resultTabYes  & \resultTabNo   & \resultTabNo   & \resultTabNo   & \resultTabYes  \\ 
  \hline \code{N+Leq+Add+Mul-ax2} & \resultTabYes  & \resultTabYes  & \resultTabYes  & \resultTabYes  & \resultTabYes  & \resultTabNo   & \resultTabYes  \\ 
  \hline \code{N+Leq+Add+Mul-ax3} & \resultTabYes  & \resultTabYes  & \resultTabYes  & \resultTabNo   & \resultTabNo   & \resultTabNo   & \resultTabYes  \\ 
  \hline \code{N+Leq+Add+Mul-ax4} & \resultTabYes  & \resultTabYes  & \resultTabYes  & \resultTabYes  & \resultTabYes  & \resultTabNo   & \resultTabYes  \\ 
  \hline \code{N+Leq+Add+Mul-ax5} & \resultTabYes  & \resultTabYes  & \resultTabYes  & \resultTabNo   & \resultTabNo   & \resultTabNo   & \resultTabYes  \\ 
  \hline \code{N+L+Pref+App-ax0} & \resultTabYes  & \resultTabYes  & \resultTabYes  & \resultTabYes  & \resultTabYes  & \resultTabYes  & \resultTabYes  \\ 
  \hline \code{N+L+Pref+App-ax1} & \resultTabYes  & \resultTabYes  & \resultTabYes  & \resultTabYes  & \resultTabYes  & \resultTabNo   & \resultTabYes  \\ 
  \hline \code{N+L+Pref+App-ax2} & \resultTabYes  & \resultTabYes  & \resultTabYes  & \resultTabNo   & \resultTabNo   & \resultTabNo   & \resultTabNo   \\ 
  \hline \code{N+L+Pref+App-ax3} & \resultTabYes  & \resultTabYes  & \resultTabYes  & \resultTabYes  & \resultTabYes  & \resultTabNo   & \resultTabYes  \\ 
  \hline \code{N+L+Pref+App-ax4} & \resultTabYes  & \resultTabYes  & \resultTabYes  & \resultTabNo   & \resultTabNo   & \resultTabNo   & \resultTabYes  \\ 

  \hline 
  \hline 

  \multicolumn{8}{|c|}{\pReflHard}\\
  \hline \code{eqRefl} & \resultTabYes  & \resultTabYes  & \resultTabYes  & \resultTabYes  & \resultTabYes  & \resultTabYes  & \resultTabYes  \\ 
  \hline \code{eqTrans} & \resultTabYes  & \resultTabYes  & \resultTabYes  & \resultTabNo   & \resultTabNo   & \resultTabNo   & \resultTabYes  \\ 
  \hline \code{excludedMiddle-0} & \resultTabYes  & \resultTabYes  & \resultTabYes  & \resultTabYes  & \resultTabYes  & \resultTabYes  & \resultTabYes  \\ 
  \hline \code{excludedMiddle-1} & \resultTabYes  & \resultTabYes  & \resultTabYes  & \resultTabYes  & \resultTabYes  & \resultTabYes  & \resultTabYes  \\ 
  \hline \code{universalInstance} & \resultTabNo   & \resultTabNo   & \resultTabYes  & \resultTabYes  & \resultTabYes  & \resultTabYes  & \resultTabYes  \\ 
  \hline \code{contraposition-0} & \resultTabYes  & \resultTabYes  & \resultTabYes  & \resultTabYes  & \resultTabYes  & \resultTabYes  & \resultTabYes  \\ 
  \hline \code{contraposition-1} & \resultTabYes  & \resultTabYes  & \resultTabYes  & \resultTabNo   & \resultTabNo   & \resultTabNo   & \resultTabYes  \\ 
  \hline \code{currying-0} & \resultTabYes  & \resultTabYes  & \resultTabYes  & \resultTabYes  & \resultTabYes  & \resultTabNo   & \resultTabYes  \\ 
  \hline \code{currying-1} & \resultTabYes  & \resultTabYes  & \resultTabYes  & \resultTabNo   & \resultTabNo   & \resultTabNo   & \resultTabYes  \\ 
  \hline \code{addGround-0} & \resultTabYes  & \resultTabYes  & \resultTabYes  & \resultTabYes  & \resultTabYes  & \resultTabYes  & \resultTabYes  \\ 
  \hline \code{addGround-1} & \resultTabYes  & \resultTabYes  & \resultTabNo   & \resultTabNo   & \resultTabNo   & \resultTabYes  & \resultTabYes  \\ 
  \hline \code{addExists} & \resultTabNo   & \resultTabNo   & \resultTabNo   & \resultTabNo   & \resultTabNo   & \resultTabNo   & \resultTabYes  \\ 
  \hline \code{existsZeroAdd} & \resultTabNo   & \resultTabNo   & \resultTabNo   & \resultTabNo   & \resultTabNo   & \resultTabNo   & \resultTabNo   \\ 
  \hline \code{mulGround} & \resultTabYes  & \resultTabYes  & \resultTabYes  & \resultTabNo   & \resultTabNo   & \resultTabNo   & \resultTabYes  \\ 
  \hline \code{mulExists} & \resultTabNo   & \resultTabNo   & \resultTabNo   & \resultTabNo   & \resultTabNo   & \resultTabNo   & \resultTabYes  \\ 
  \hline \code{existsZeroMul} & \resultTabNo   & \resultTabNo   & \resultTabNo   & \resultTabNo   & \resultTabNo   & \resultTabNo   & \resultTabNo   \\ 
  \hline \code{appendGround-0} & \resultTabYes  & \resultTabYes  & \resultTabYes  & \resultTabYes  & \resultTabYes  & \resultTabYes  & \resultTabYes  \\ 
  \hline \code{appendGround-1} & \resultTabYes  & \resultTabYes  & \resultTabYes  & \resultTabNo   & \resultTabNo   & \resultTabYes  & \resultTabYes  \\ 
  \hline \code{appendExists} & \resultTabNo   & \resultTabNo   & \resultTabNo   & \resultTabNo   & \resultTabNo   & \resultTabNo   & \resultTabYes  \\ 
  \hline \code{existsNil} & \resultTabNo   & \resultTabNo   & \resultTabYes  & \resultTabYes  & \resultTabYes  & \resultTabNo   & \resultTabYes  \\ 
  \hline 
  \end{tabular}
}

%% file: conclusion.tex

It is mathematical practice to define infinite sets of axioms as schemes of formulas. Alas these schemes of axioms are not part of standard input syntax of today's theorem proves. 
In order to circumvent this shortcoming, we developed a method to express these schematic definitions in the language of first-order logic by means of a conservative extension, which we called the reflective extension of a theory. 
We showed that this reflective extension is indeed a conservative extension of the base theory. It contains a truth predicate
which allows  us  to quantify over formulas within the language of first-order logic. 

We replaced the first-order induction scheme of \PA{} by the axioms needed for the reflective extension and a single additional axiom, called the reflective induction axiom. We  proved that the resulting theory is indeed a conservative extension of \PA. 
Further,  we demonstrated how to replace the induction scheme of a theory with arbitrary inductive datatypes. This kind of conservative extension is what we called the reflective inductive extension.

Our experiments show that reasoning in the reflective extension of a theory is hard for modern theorem provers, even for very simple problems. Despite the poor performance in general, we have a positive result serving as a proof of concept of our method, namely that the SMT-solver \solver{Z3}, which does not support induction natively was able to solve problems that require inductive reasoning. 

Investigating our encoding in relation with the proof systems supported by the Dedukti framework~\cite{Dedukti} is an interesting line for further work. Further we are interested to explore which different proof search heuristics can be used to make our technique feasible for practical applications.